\documentclass[aps,prl,preprint,groupedaddress]{revtex4-1}

\usepackage{graphicx}
\usepackage{dcolumn}
\usepackage{bm}
\usepackage{color}

\usepackage{comment}
\usepackage[geometry]{ifsym}
\usepackage{booktabs}
\usepackage{amsthm}
\usepackage{amsmath}
\newtheorem{definition}{Definition}
\theoremstyle{remark}
\newtheorem{remark}{Remark}
\newtheorem{theorem}{Theorem}
\newtheorem{proposition}[theorem]{Proposition}
\newtheorem{lemma}[theorem]{Lemma}
\newtheorem{corollary}[theorem]{Corollary}

\newtheorem{conjecture}{Conjecture}
\usepackage{amssymb}
\usepackage{booktabs}
\usepackage{amsthm}

\usepackage{enumerate}
\usepackage{multirow}
\usepackage{caption}
\newtheorem{observation}{Observation}
\begin{document}


\title{    Bulk    Phase    Transition    and    Edge    Behavior    in    Temporally    Correlated    Random    Matrices
}


\author{Masato        Hisakado}
\email{hisakadom@yahoo.co.jp}        
\affiliation{
*Kanazawa        university,        Kakumamachi,        Kanazawa,        Ishikawa        920-1192,        Japan}        

\author{Takuya        Kaneko}
\email{tkaneko@icu.ac.jp}
\affiliation{
\dag        
International        Christian                University        \\
Osawa        3-10-2,        Mitaka,        Tokyo        181-8585,        Japan}


\date{\today}

\begin{abstract}
We  study  long-range  correlated  Wigner-type  matrices  built  from  row-independent  stationary
Gaussian  sequences.  For  exponentially  decaying  (AR(1))  correlations,  the  bulk  spectral  density
deforms  from  the  semicircle  law  via  an  explicit  combinatorial  “hub”  mechanism,  yet  we  verify  the
flatness  and  decay  hypotheses  of  the  matrix-Dyson-equation  framework  (MDE)  \cite{AEKS2020},  with  numerical  evidence supporting Tracy–Widom edge universality
for  every  fixed  $\rho  <  1$  of  the  exponential  decay  correlations;  the  degenerate  limit  $\rho\to1^-  $
reduces  to  a  symmetrized  Volterra  operator,  connecting  to  the  singular-value  cascade  identified  in
a  companion  BBP  analysis.
For  power-law  correlations  $dt\sim  t^{-\gamma}$,    we  identify  $\gamma_c  =  1/2$  
as  the  critical  point  for  divergence  of  the  bulk  fourth-moment,  while  $\gamma  =  1$  marks  the  breakdown  of  the  flatness  condition  governing  the  MDE  edge  analysis.  We  prove  the  fourth-moment  transition  exactly  and  find  numerically  that  the  self-consistent  edge  varies  smoothly  across  $\gamma  =  1$,  with  no  evidence  of  a  kink  or  discontinuity.

\end{abstract}
\maketitle
\section{I.        Introduction}

Random        matrix        theory        (RMT)        provides        a        universal        framework        for        describing        the        spectral        behavior        of        large        complex        systems        \cite{Meh}.        Depending        on        the        structure        of        the        underlying        random        matrix,        distinct        limiting        eigenvalue        distributions        arise.        For        instance,        the        Marchenko–Pastur        distribution        (MPD)        governs        the        spectrum        of        Wishart        matrices,        while        the        Wigner        semicircle        law        characterizes        the        spectrum        of        symmetric        random        matrices        with        independent        entries.        Deviations        from        these        classical        laws        often        signal        additional        structures        such        as        correlations        or        heavy-tailed        statistics,        and        have        been        widely        studied        in        physics,        mathematics,        and        applications        ranging        from        nuclear        spectra        to        network        theory        \cite{Pot,Pot3,Pot2,Fra,RM,AI}.

We        study        a        family        of        long-range        temporally        correlated        Wigner-type
random        matrices        constructed        by        assigning        to        each        row        of        a        symmetric
$N\times        N$        matrix        an        independent        stationary        Gaussian        sequence.        
So
that        correlation        is        confined        within        rows        while        remaining        exactly
absent        across        rows.        For        exponentially        decaying        (AR(1))        row
correlations,        we        show        algebraically        that        the        ensemble        admits        a        causal
linear-filter        representation        and        that        its        fourth-moment        deviates        from
the        semicircle        value        through        a        combinatorially        explicit        ``hub''
mechanism,        so        that        the        limiting        spectral        density        is        a        $\rho$-dependent
deformation        of        the        semicircle        law        rather        than        the        semicircle        law
itself.        We        verify        the        hypotheses        (A),        (B),        (CD),        and        (E)        of        the
matrix-Dyson-equation        (MDE)        framework        for        correlated        Wigner-type        matrices        \cite{AEKS2020,AjankiErdosKruger2017,ErdosKrugerSchroder2019}        
        explicitly
for        this        ensemble,        with        the        flatness        constants        degenerating        as
$\rho\to1^-$,        and        give        numerical        evidence        (moment        convergence,        edge
scaling        exponents,        and        skewness/kurtosis        of        $\lambda_{\max}$)        that,
despite        the        deformed        bulk,        Tracy--Widom        (TW)        edge        universality        \cite{tra94,tra96}                persists        for
every        fixed        $\rho<1$.        

In        the        degenerate        limit        $\rho\to1$,        the        ensemble
reduces        exactly        to        a        symmetrized        discrete        Volterra        operator        with
random        (signed)        weights        and        no        background        noise;        the        bulk        spectral
measure        collapses        to        $\delta_0$        while        the        extreme        eigenvalues        are
numerically        consistent,        up        to        a        $\sqrt        N$        rescaling,        with        the
Volterra-operator        singular-value        cascade        $1/(\pi(k-\tfrac12))$
identified        in        the        companion        analysis        of                Baik--Ben        Arous--P\'ech\'e        (BBP)        \cite{bbp}        transitions        for
Volterra-type        Wigner        matrices,        providing        a        structural        bridge        between
the        two        constructions.
For        power-law        row        correlations        $d_t\sim        t^{-\gamma}$,        we        identify        
distinct        critical        exponents        governing        different        spectral        observables:
$\gamma_c=1/2$,        below        which        even        the        bulk        moments        diverge,        and
$\gamma_c=1$,        coinciding        with        the        breakdown        of        the        boundedness
(flatness)        condition        required        for        the        MDE        approach.
We        prove                the        phase        transition        of        the        fourth-moment        strictly.
Numerically,        the        intermediate        regime        $1/2<\gamma\le1$        exhibits      that  the        bulk        spectral        moments        remain
finite        while        the        largest        eigenvalue        is  finite.
Therefore,  no  evidence  for  a  critical  point  is  found  at  $\gamma=1$.

In  the  range  $1/2<\gamma<1$,  the  auxiliary  measure  $H_{\gamma}$
  develops  a  hierarchy  of  finite-moment  thresholds,  $\gamma_c(n)=(n-1)/n$,
as  detailed  in  Appendix  C.  These  thresholds  govern  the  structure  of  the  asymptotic  moment  expansion,  but  do  not  correspond  to  independent  edge  phase  transitions.


The        remainder        of        this        paper        is        organized        as        follows.
In        Section        II,        we        introduce                        Wigner        matrix        with        the        exponential        decay        correlation.
In        Section        III,        we        consider        the        strong        correlation        limit                of        the        exponential                decay        case        and        the        relation        to        the        Volterra        operator.
In        Section        IV,        we        discuss    the    fourth-moment    transition    and        the        edge        of        Wigner        matrix        with        the        power        decay        correlation.
Finally,        the        conclusions        are        presented        in        Section        V.

\section{II        Edge        Distribution        of        the        Exponential        decay        case}
\subsection{A.        The        Row-Independent        AR(1)        Ensemble}
\label{subsec:row-ar1}

We        now        consider        a        variant        of        the        folded        Wigner        matrix        in        which        temporal
correlation        is        confined        strictly        within        each        row,        so        that        no        correlation
survives        across        the        row        boundary.

\begin{definition}[Row-independent        AR(1)        matrix]
\label{def:row-ar1}
Fix        $N\in\mathbb        N$,        $\sigma>0$,        and        $\rho\in(-1,1)$.        For        each
$i=1,\dots,N$        let        $\{\eta_i(j)\}_{j=i}^{N}$        be        a        stationary        Gaussian
AR(1)        process,
\[
\eta_i(i)\sim\mathcal        N(0,1),\qquad
\eta_i(j)=\rho\,\eta_i(j-1)+\sqrt{1-\rho^2}\,\varepsilon_i(j),
\qquad        j>i,
\]
with        $\{\varepsilon_i(j)\}$        jointly        independent        standard        Gaussians,        and
the        $N$        processes        $\{\eta_i(\cdot)\}_{i=1}^N$        mutually        independent.
Define        the        symmetric        matrix        $S=S^{(N)}(\rho)$        by
\[
S_{ij}=S_{ji}:=\frac{\sigma}{\sqrt        N}\,\eta_i(j),\qquad        i\le        j.
\]
\end{definition}

By        construction,
\begin{equation}
\label{eq:cov-structure}
\operatorname{Cov}(S_{ab},S_{cd})
=\frac{\sigma^2}{N}\,
\delta_{\min(a,b),\,\min(c,d)}\;
\rho^{\,|\max(a,b)-\max(c,d)|}.
\end{equation}
In        particular        entries        belonging        to        different        rows        are        exactly
independent;        only        entries        within        a        common        row        are        correlated,        with
exponential        (AR(1))        decay        in        the        column        separation.

\begin{proposition}[Operator        representation]
\label{prop:EF}
There        is        a        family        of        i.i.d.\        standard        Gaussian        variables
$\{E_{i,k}\}_{i\le        k\le        N}$        and        a        row-dependent        causal        filter
\[
F^{(i)}_{k,j}=
\begin{cases}
\rho^{\,j-i},        &        k=i,\\[2pt]
\sqrt{1-\rho^2}\,\rho^{\,j-k},        &        i<k\le        j,
\end{cases}
\qquad        j\ge        i,
\]
such        that
\[
S_{ij}=\frac{\sigma}{\sqrt        N}\sum_{k=i}^{j}E_{i,k}\,F^{(i)}_{k,j},
\qquad        i\le        j.
\]
The        filter        has        $\ell^1$-summable        coefficients,
$\sum_{l\ge0}\sup_i|F^{(i)}_{i+l,\,\cdot}|<\infty$,        uniformly        in        $N$.
\end{proposition}

\begin{remark}
\label{r1}
The        row        dependence        of        $F^{(i)}$        enters        only        through        the        first
(``seed'')        column;        away        from        $k=i$        the        filter        is        a        fixed        Toeplitz
operator        with        symbol
$\hat        f(\theta)=1+\sqrt{1-\rho^2}\sum_{l\ge1}\rho^le^{-il\theta}$,
bounded        and        bounded        away        from        $0$        for        every        fixed        $\rho\in(-1,1)$.
It        is        the        different        point        from        \cite{Hisakado2026}        where        we        use        one        time        series.

\end{remark}

\subsection{B.        Deviation        from        the        Semicircle        Law}

\begin{proposition}[Non-semicircular        bulk]
\label{prop:mu4}
Let        $\mu_{2k}:=\lim_{N\to\infty}\frac1N\,\mathbb        E\operatorname{Tr}(S^{2k})$.
Then
\[
\mu_4(\rho)        >        2\sigma^4=\mu_4^{\mathrm{sc}},\qquad
\mu_4(\rho)\longrightarrow        2\sigma^4        \text{        as        }\rho\to0,\qquad
\mu_4(\rho)\longrightarrow\infty        \text{        as        }\rho\to1^-.
\]
Consequently        the        limiting        spectral        distribution        of        $S$        is        not
the        semicircle        law        for        $\rho\neq0$.
\end{proposition}
\begin{proof}
The        proof        is        in        \cite{Hisakado2026}.
\end{proof}
\begin{proposition}[Bulk    fourth-moment    divergence    and    critical    exponent,    AR(1)    case]
\label{prop:AR1-bulk-exponent}
Under    the    general    hub-mechanism    formula    of    Remark~\ref{r8}    (applicable    to    any
row-independent,    Toeplitz    covariance    sequence    $\{d_t\}$),    the    AR(1)    ensemble
with    $d_t=\rho^t$    satisfies    the    exact    closed    form
\begin{equation}
        \mu_4(\rho)    =    2\sigma^4    +    \frac{4}{3}\sigma^4\,\frac{\rho^2}{1-\rho^2},
        \qquad    \rho    \in    (-1,1),
        \label{e2}
\end{equation}
consistent    with    $\mu_4(\rho)\to    2\sigma^4$    as    $\rho\to0$    (Proposition~\ref{prop:mu4})    \cite{Hisakado2026}.
As    $\rho    \to    1^-$,    writing    $\varepsilon    :=    1-\rho$,
\begin{equation}
        \mu_4(\rho)    -    2\sigma^4    \;\sim\;    \frac{2}{3}\sigma^4\,\varepsilon^{-1},
        \qquad    \varepsilon    \to    0^+,
\end{equation}
a    simple    pole    with    critical    exponent    $\psi    =    1$    ---    identical    in    both
exponent    and    residue    ($2/3$,    with    $\sigma=1$)    to    the    power-law
bulk    transition    at    $\gamma_c=1/2$    (Proposition~\ref{prop:super-critical-exponent}).    This    coincidence
reflects    a    common    origin:    in    both    cases    the    divergence    of
$\sum_t    d_t^2$    is    a    simple    pole    in    the    natural    distance-to-criticality
variable.
\end{proposition}
\subsection{C.        Edge        Universality        via        the        Matrix        Dyson        Equation}

Although        the        bulk        deviates        from        the        semicircle        law,        TW
universality        does        not        require        a        semicircular        bulk:        it        requires
only        that        the        matrix        falls        under        the        Matrix        Dyson        Equation        (MDE)
framework        and        that        the        resulting        self-consistent        edge        is
regular        (square-root        vanishing).        We        verify        the        relevant
hypotheses        of
\cite{AEKS2020}        for        the        ensemble        of
Definition~\ref{def:row-ar1}.

\begin{proposition}[Verification        of        Assumptions        (A),(B),(CD),(E),        \cite{AEKS2020}        ]
\label{prop:assumptions}
Let        $W=S$        as        in        Definition~\ref{def:row-ar1},        with        $S[T]:=\mathbb
E[WTW]$.        Then,        uniformly        in        $N$:
\begin{enumerate}
\item[(A)]        $A=\mathbb        E        W=0$,        so        Assumption        (A)        holds        trivially.
\item[(B)]        All        entries        are        Gaussian        with        variance        $\sigma^2/N$,        so        all
                moments        of        $\sqrt        N        W_{ab}$        are        bounded        uniformly.
\item[(CD)]        For        entries        in        different        rows,        $\operatorname{Cov}=0$
                identically;        within        a        row,        $\operatorname{Cov}$        decays        as
                $\rho^{|\cdot|}$,        exponentially        faster        than        the        required        polynomial
                rate        $s>1/2$        in        the        metric        on        $[N]^2$.
\item[(E)]        (Flatness)        For        any        positive        semi-definite        $T$,
\[
S[T]_{aa}=\frac{\sigma^2}{N}\Big[\sum_{c<a}T_{cc}
+\sum_{c,d>a}\rho^{|c-d|}T_{cd}\Big].
\]
Since        the        AR(1)        kernel        $(\rho^{|c-d|})_{c,d>a}$        has        spectrum        contained
in        $\big[\tfrac{1-\rho}{1+\rho},\tfrac{1+\rho}{1-\rho}\big]$        (the
essential        range        of        its        symbol),        we        obtain
\[
\frac{1-\rho}{1+\rho}\,\langle        T\rangle
\        \le\        S[T]_{aa}\        \le\
\frac{1+\rho}{1-\rho}\,\langle        T\rangle,
\]
so        Assumption        (E)        holds        with        constants        $c(\rho)=\tfrac{1-\rho}{1+\rho}$,
$C(\rho)=\tfrac{1+\rho}{1-\rho}$,        uniform        in        $N$        for        each        fixed
$\rho\in(-1,1)$        (degenerating        as        $\rho\to1^-$).
\end{enumerate}
\end{proposition}

\begin{remark}
\label{r2}
Assumption        (F)        (fullness)        is        expected        to        hold        by        the        same
non-degeneracy        ($\rho<1$)        but        has        not        been        verified        in        detail.
Assumption        (G)        (boundedness        of        $M(z)$        near        the        edge)        is,        per
\cite{AEKS2020},        not        automatic        from        (A)--(E)        and
must        be        checked        model        by        model.        Fixed-point        iteration        of        the        MDE,
$M=-(zI+S[M])^{-1}$        (Monte        Carlo        evaluation        of        $S[\cdot]$)        shows
$\|M(\tau_0+i\eta)\|$        remaining        bounded        as        $\eta\downarrow0$        for
$\rho\in\{0,0.3,0.5\}$        down        to        $\eta=3\times10^{-3}$,        with        the
self-consistent        density        $\varrho(\tau_0+i\eta)\to0$,        consistent        with        a
regular              edge.        The        check        becomes        numerically        unreliable
for        $\rho\gtrsim0.7$,        consistent        with        the        degeneration        of        the
Assumption        (E)        constants        as        $\rho\to1$.
\end{remark}
\begin{definition}
\label{def3}
Let    $H$    be    a    probability    measure    on    $[0,\sigma_1]$    for    some    $\sigma_1    
<\infty$,    arising    as    the    pushforward    of    the    uniform    measure    on    $[0,2\pi)$
under    a    bounded,    nonnegative    spectral    density    $f(\theta)$,    with
\begin{equation}
        \bar    t    :=    \int    t    \,    dH(t)    =    1,    \qquad
        \sigma_1    :=    \operatorname*{ess\,sup}    f    =    \sup(\operatorname{supp}    H)    <    \infty.
\end{equation}
\end{definition}
    The    self-consistent    edge
$(\tau_0,\xi_+)$    is    characterized,    via    the    substitution    $a    =    -z/m$,    $\tau_0    =
-a^*    m^*$,    by
\begin{equation}
        I(a)    :=    \int    \frac{dH(t)}{a-t},    \qquad
        K_1(a)    :=    \int    \frac{t\,dH(t)}{(a-t)^2},    \qquad
        I(a^*)    =    K_1(a^*),
\end{equation}
with    $m^*    =    -\sqrt{I(a^*)}$    the    branch    relevant    to    the    right    edge;    existence
of    a    solution    $a^*    >    \sigma_1$    is    guaranteed    by    the    general    theory    of
[1,\,9]    under    Assumptions    (A)--(E).

\begin{conjecture}[Edge        universality        for        the        exponential        decay        case]
\label{conj:TW}
For every fixed $\rho \in (-1,1)$, the fluctuations of $\lambda_{\max}(S)$
around the self-consistent edge $\tau_0(\rho)$, suitably rescaled,
converge in distribution to the Tracy–Widom distribution for the GOE ($\mathrm{TW}_1$  law):
\[
\lambda_{\max}(S) - \tau_0(\rho) \;\xrightarrow{d}\; \mathrm{TW}_1
\quad \text{(after an $N$-dependent, $\rho$-dependent rescaling).}
\]
\end{conjecture}

D.        Explicit        Bulk        Self-Consistent        Equation

While        Proposition~\ref{prop:assumptions}        verifies        the        qualitative        hypotheses        of        the
MDE        framework,        and        Remark~\ref{r2}        reports        only
a        Monte        Carlo        boundedness        check        of        $M(z)$        near        the        edge,        the
row-independent,        Toeplitz        structure        of        the        ensemble        of
Definition~\ref{def:row-ar1}        in        fact        permits        the        self-consistent        
equation        to        be        written        down        explicitly,        in        closed        scalar        form,
rather        than        evaluated        only        as        a        matrix-valued        fixed        point.

\begin{proposition}        (Scalar        self-consistent        equation        for        the        bulk
density).        Let        $\hat        f(\theta)$        denote        the        causal        filter        symbol
of        Remark~\ref{r1},        and        define        the        (real,        nonnegative)        power        spectral
density        of        the        row        covariance,
\begin{equation}
f(\theta)        \;:=\;        |\hat        f(\theta)|^2        \;=\;        \frac{1-\rho^2}{1-2\rho\cos\theta+\rho^2},
\qquad        \theta\in[0,2\pi).
\label{eq:ftheta}
\end{equation}
Then        the        Stieltjes        transform        $m(z)$        of        the        limiting        spectral
measure        of        $S$        satisfies        the        scalar        self-consistent        equation
\begin{equation}
m(z)        \;=\;        \int_0^{2\pi}\frac{d\theta}{2\pi}\,
\frac{1}{-z-f(\theta)\,m(z)},
\qquad        z\in\mathbb{C}^+,
\label{eq:sce}
\end{equation}
and        the        limiting        spectral        density        is        recovered        as
$\varrho(x)=\tfrac{1}{\pi}\lim_{\eta\downarrow0}\operatorname{Im}m(x+i\eta)$.
For        $\rho=0$,        $f\equiv1$        and        Eq.        \eqref{eq:sce}        reduces        to        the
ordinary        semicircle        self-consistent        equation
$m=1/(-z-m)$,        consistent        with        Proposition~\ref{prop:mu4}.
\label{p4}
\end{proposition}
\begin{proof}
By        Proposition~\ref{prop:assumptions},        the        ensemble        of        Definition~        \ref{def:row-ar1}        satisfies        Assumptions
(A)--(E)        with        self-consistent        covariance        operator        $S[T]_{aa}        =        \frac{\sigma^2}{N}
\sum_{c,d}        (\text{Toeplitz        kernel})_{cd}        T_{cd}$,        whose        associated        symbol        is
$f(\theta)        =        |\hat        f(\theta)|^2$        (Remark~\ref{r1}).        By        the        universality        of
the        MDE        \cite{AEKS2020,        AjankiErdosKruger2017}                which        states        that        the        self-consistent
density        $\varrho$        associated        to        a        Wigner-type        ensemble        depends        on        the        ensemble
only        through        the        operator        $S[\cdot]$,        and        not        on        finer        details        of        the        joint
law        of        the        entries        ---        the        limiting        spectral        density        of        $S$        coincides        with
that        of        any        Wigner-type        ensemble        sharing        the        same        operator        $S[\cdot]$.
In        particular,        it        coincides        with        that        of        $\Sigma^{1/2}        G        \Sigma^{1/2}$,        where
$G$        is        a        standard        GOE        matrix        and        $\Sigma$        is        the        Toeplitz        covariance        operator
with        symbol        $f(\theta)$:        since        $G$        is        orthogonally        invariant,        $\Sigma$        and
$OGO^T$        ($O$        Haar-distributed,        independent        of        $G$'s        eigenvalues)        are
asymptotically        free,        so        the        limiting        spectral        distribution        of
$\Sigma^{1/2}G\Sigma^{1/2}$        is        the        free        multiplicative        convolution        of        the
semicircle        law        with        the        law        of        $f(\theta)$        under        $\theta$        uniform,        whose
Stieltjes        transform        satisfies        Eq.~(\ref{eq:sce}).
\end{proof}

\begin{remark}        (Consistency        with        Assumption        (E)).        The        essential        range
of        $f(\theta)$        is        exactly        the        interval
$\bigl[\tfrac{1-\rho}{1+\rho},\,\tfrac{1+\rho}{1-\rho}\bigr]$,
i.e.\        precisely        the        flatness        bounds        $c(\rho)$,        $C(\rho)$        of
Proposition~\ref{prop:assumptions}        (E).        Equation~\eqref{eq:sce}        thus        makes        explicit,
at        the        level        of        a        single        scalar        fixed-point        equation        rather        than
an        operator-valued        one,        exactly        the        object        whose        boundedness
Assumption~(E)        controls;        the        degeneration        $c(\rho)\to0$,
$C(\rho)\to\infty$        as        $\rho\to1^-$        (Proposition~\ref{prop:E-breakdown})        corresponds
to        $f(\theta)$        developing        an        unbounded        peak        at        $\theta=0$        in
Eq.        \eqref{eq:ftheta}.
\end{remark}

Equation~\eqref{eq:sce}        is        solved        numerically        by        fixed-point
iteration        in        $m$        at        $z=x+i\eta$        for        a        grid        of        $x$        and        small
$\eta>0$,        using        continuation        in        $x$        for        stability        (the        solution
at        each        $x$        seeds        the        initial        guess        at        the        next).        This        is        a
direct,        non-Monte-Carlo        alternative        to        the        boundedness        check        of
Remark~\ref{r2},        and        additionally        yields        the        bulk        density        itself,        not
merely        a        boundedness        certificate.

Figure~\ref{def:row-ar1}        shows        the        corresponding        densities        and        residuals
directly.        The        self-consistent        theory        tracks        the        empirical
histogram        closely        at        both        values        of        $\rho$,        including        the
excess        central        peak        and        the        depleted        shoulders        that
distinguish        $\varrho$        from        the        semicircle        law        (Proposition~\ref{prop:mu4});
the        semicircle        curve,        by        contrast,        shows        systematic,
$O(10^{-1})$        residuals        of        a        characteristic        sign        pattern        (too
low        at        the        center,        too        high        on        the        shoulders)        that        grow        with
$\rho$.

\begin{figure}[h]
\centering
\includegraphics[width=\textwidth]{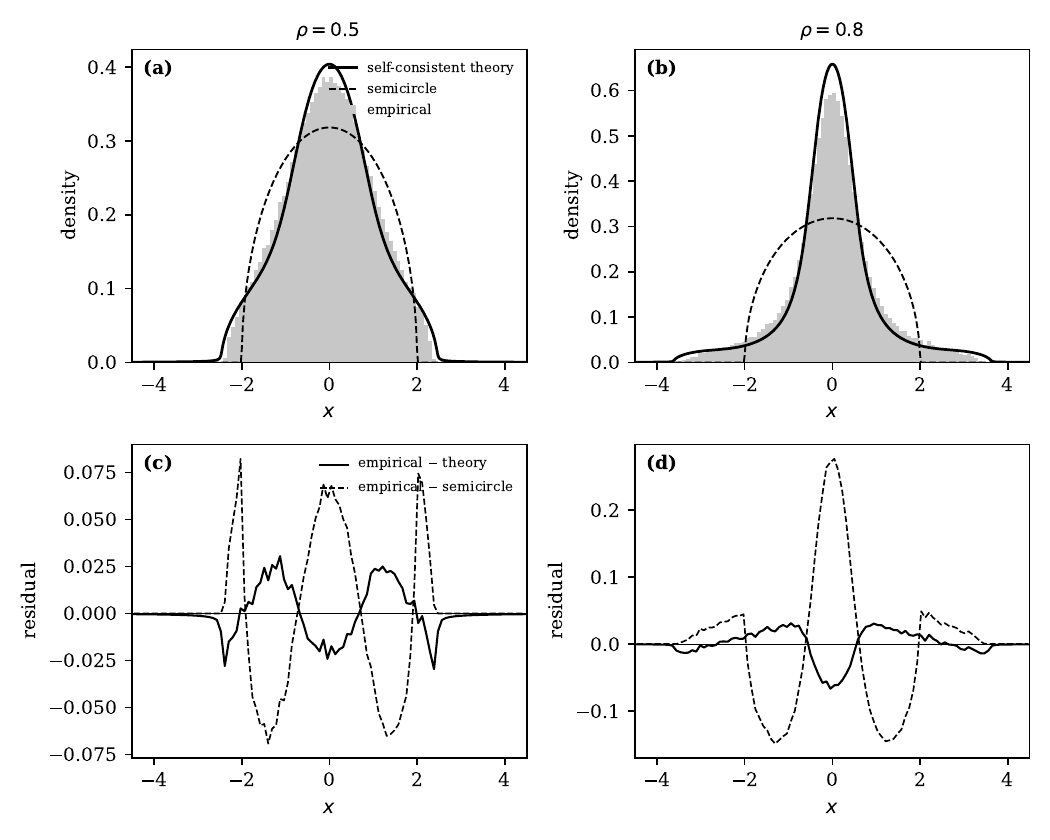}
\caption{(a,b)        Bulk        spectral        density        of        $S$        at        $\rho=0.5$        and
$\rho=0.8$:        empirical        histogram        (gray),        
Eq.        \eqref{eq:sce}        (solid        black),        and        the        semicircle        law        (dashed
black).        (c,d)        Residuals        of        the        theory        and        of        the        semicircle        law
relative        to        the        empirical        histogram,        for        the        same        two        values        of
$\rho$.        $N=800$,        $25$        realizations,        $\eta=0.02$.}
\label{fig:bulk_sce}
\end{figure}

We        regard        Eq.        \eqref{eq:sce},        together        with        the        numerical
agreement        of        
Fig.~\ref{fig:bulk_sce},        as        an        explicit        closed-form
replacement        for        the        Monte        Carlo        boundedness        check        of        Remark~\ref{r2}.

We        emphasize        that        $f(\theta)$        alone        does        not        determine
Eq.~\eqref{eq:sce}:        $f(\theta)$        is        exactly        the        limiting
spectral        measure        of        the        Toeplitz        covariance        operator        $\Sigma_\rho
=        (\rho^{|j-k|})_{j,k}$        (Szegő's        theorem),        and        depends        only        on        the
row        covariance        $\{\rho^{|t|}\}$;        but        Eq.~\eqref{eq:sce}        itself        is        a
specific        self-consistency        relation        ---        the        free        multiplicative
convolution        of        the        semicircle        law        with        the        law        of        $f(\theta)$
($\theta$        uniform)        ---        that        additionally        presupposes        the
row-independence        structure        of        Definition~\ref{def:row-ar1}        (Assumption~(CD):        zero
covariance        across        distinct        rows).        Given        only        the        spectral        measure
$\mu$        of        $f(\theta)$,        without        this        independence        structure,        the
map        from        $\mu$        to        the        limiting        spectral        density        $\varrho$        is        not
determined;        a        different        matrix        construction        sharing        the        same        row
covariance        but        correlating        entries        across        rows        would        in        general
satisfy        a        different        self-consistent        equation,        or        none        of        this
scalar        form        at        all.        Equation~\eqref{eq:sce}        should        therefore        be
read        as        jointly        encoding        two        independent        ingredients:        the
Toeplitz        spectrum        $f(\theta)$        (which        fixes        the        coefficients        of
the        equation)        and        the        row-independent        ensemble        structure        of
Definition~\ref{def:row-ar1}        (which        fixes        the        form        of        the        equation).
\begin{remark}[Finiteness    of    the    edge,    AR(1)    case]\label{rem:finite-edge-AR1}
For    every    fixed    $\rho    \in    (-1,1)$,    the    self-consistent    edge    $\tau_0(\rho)$    is
finite;( see Appendix~A Corollary \ref{cor:finite-AR1}).
for    a    two-sided    quantitative    bound    in    terms    of
$\sigma_1(\rho) =  \operatorname*{ess\,sup}    f    =    f(0)     =    \tfrac{1+\rho}{1-\rho}$.
\end{remark}


\begin{proposition}[Divergence    of    the    self-consistent    edge    as    $\rho    \to    1^-$]
\label{prop:AR1-edge-divergence}
As    $\rho    \to    1^-$,
\begin{equation}
        \sqrt{\sigma_1(\rho)/2}    \;<\;    \tau_0(\rho)    \;<\;    2\sqrt{2\sigma_1(\rho)},
        \qquad    \sigma_1(\rho)    =    \frac{1+\rho}{1-\rho},
\end{equation}
so    in    particular    $\tau_0(\rho)    =    \Theta(\sigma_1(\rho)^{1/2})    \to    \infty$.
Since    $\sigma_1(\rho)    \sim    2/(1-\rho)$    as    $\rho    \to    1^-$,    this    gives    the
explicit    rate
\begin{equation}
        \tau_0(\rho)    =    \Theta\!\left((1-\rho)^{-1/2}\right),    \qquad    \rho    \to    1^-.
\end{equation}
\end{proposition}

\begin{proof}
Immediate    from    Corollary~\ref{cor:finite-AR1}    (which    holds    for    every    fixed
$\rho<1$)    by    letting    $\rho    \to    1^-$:    both    the    lower    and    upper    bounds    diverge
since    $\sigma_1(\rho)    \to    \infty$.
\end{proof}

\begin{remark}
This    divergence    concerns    the    self-consistent    edge    $\tau_0(\rho)$    in    the    order
of    limits    $N    \to    \infty$    (for    fixed    $\rho<1$)    followed    by    $\rho    \to    1^-$;    it
is    a    distinct    statement    from    the    growth
$\lambda_{\max}(S)    =    \Theta(\sqrt{N})$    established    numerically    at    $\rho=1$
exactly    in    Observation~\ref{o1}    (Section~III),    where    Assumption~(E)
fails    and    the    MDE    framework    does    not    apply    at    all.    The    two
results    are    consistent    in    spirit    ---    both    reflect    the    same    breakdown    of
flatness    documented    in    Proposition~\ref{prop:degeneration}    ---    but    neither    implies    the    other,    since
the    orders    of    the    $N\to\infty$    and    $\rho\to    1$    limits    are    reversed.
\end{remark}
\begin{remark}
The    AR(1)    ensemble's    Poisson-kernel    structure    in    fact    permits    the
self-consistent    edge    equation    to    be    solved    in    closed    algebraic    form,
yielding    the    exact    leading    asymptotic    $\tau_0(\rho)    \sim    \sqrt2\,(1-\rho)^{-1/2}$
as    $\rho\to1^-$;    see    Appendix~B.
\end{remark}
\subsection{E.        Numerical        Evidence        of        TW                of        exponential        decay        case}
\label{subsec:numerics}

        Table
\ref{tab:edge}        reports        the        extrapolated        edge        location        $\tau_0(\rho)$
and        the        fitted        exponent        of        $\operatorname{std}(\lambda_{\max})\sim
N^{-\beta}$,        consistent        with        $\beta=2/3$        for        $\rho\le0.5$.        Skewness        and
excess        kurtosis        of        $\lambda_{\max}$        at        $N=512$--$1024$        are        compatible,
within        sampling        error,        with               $\mathrm{TW}_1$        reference        values
$(0.294,\,0.165)$.

\begin{center}
\begin{tabular}{c|cc}
\toprule
$\rho$        &        extrapolated        edge        $\tau_0(\rho)$        &        fitted        $\beta$        (std\,$\sim        N^{-\beta}$)\\
\hline
\midrule
0.0        &        1.99        &        0.72\\
0.3        &        2.13        &        0.67\\
0.5        &        2.36        &        0.70\\
0.7        &        2.89        &        0.59\\
\bottomrule
\end{tabular}
\captionof{table}{Edge        location        and        fluctuation        exponent
(TW        prediction:        $\beta=2/3$).}
\label{tab:edge}
\end{center}

\section{III                Strong        Correlation        Limit        of        exponential        decay        case}
\label{subsec:rho-to-1}

\subsection{A.        Volterra        Operator}
We        now        examine        the        limit        $\rho\to1^-$        of        the        ensemble        of
Definition~\ref{def:row-ar1}.        This        limit        is        singular:        it        lies        outside
the        regime        covered        by        Proposition~\ref{prop:assumptions},        and        connects
the        present        construction        directly        to        the        Volterra-operator        cascade
studied        in        the        companion        paper        on        multi-critical        BBP        transitions        \cite{Hisakado2026_BBP}.

\begin{proposition}[Algebraic        degeneration        at        $\rho=1$]
\label{prop:degeneration}
As        $\rho\to1^-$,        the        innovation        term        of        each        row's        AR(1)        process
vanishes        ($\sqrt{1-\rho^2}\to0$),        so        that        $\eta_i(j)\to\eta_i(i)=:\eta_i$
for        all        $j\ge        i$.        Consequently
\[
S_{ij}\        \xrightarrow{\rho\to1^-}\        \frac{\sigma}{\sqrt        N}\,\eta_{\min(i,j)},
\qquad        \eta_1,\dots,\eta_N\        \text{i.i.d.\        }N(0,1).
\]
\end{proposition}

\begin{remark}
This        is        exactly        the        degenerate        case        $\eta_i(j)\equiv\eta_i$
(constant        along        each        row),        so        the        family        of        Definition~\ref{def:row-ar1}
interpolates        continuously        between        the        Wigner-type        ensemble        ($\rho=0$)
and        this        fully        rank-structured        limit        ($\rho=1$).
\end{remark}

\begin{proposition}[Breakdown        of        Assumption        (E)]
\label{prop:E-breakdown}
The        flatness        constants        of        Proposition~\ref{prop:assumptions}        satisfy
$c(\rho)=\frac{1-\rho}{1+\rho}\to0$        and
$C(\rho)=\frac{1+\rho}{1-\rho}\to\infty$        as        $\rho\to1^-$.        Hence
Assumption        (E)        fails        in        the        limit,        and        the        MDE
framework        of        Section        II        does        not        apply        at        $\rho=1$;        a
different                description        is        required.
\end{proposition}

\begin{proposition}[Noiseless        Volterra        representation]
\label{prop:LDL}
Let        $w_k:=\sum_{i\ge        k}e_i\in\mathbb        R^N$        and        let
$L=[w_1\,|\,w_2\,|\cdots|\,w_N]$        be        the        $N\times        N$        lower-triangular
matrix        of        ones,        $L_{ik}=\mathbb        1[k\le        i]$        (the        discrete        Volterra
operator).        Set        $\delta_k:=\eta_k-\eta_{k-1}$        ($\eta_0:=0$).        Then        the
matrix        $M_{ij}:=\eta_{\min(i,j)}$        of        Proposition~\ref{prop:degeneration}
satisfies        the        exact        identity
\[
M        \;=\;        L\,\operatorname{diag}(\delta)\,L^{T}
\;=\;\sum_{k=1}^{N}\delta_k\,w_kw_k^{T}.
\]
Equivalently,        $S=M/\sqrt        N$        (with        $\sigma=1$)        is        the        sandwiched
Volterra        ensemble        $S=\dfrac{1}{\sqrt        N}L\,\operatorname{diag}(\delta)\,L^{T}$,
i.e.\        the        companion-paper        construction
$S_{ij}=bZ_{\max(i,j)}/\sqrt        N+\sigma\,(\text{Wigner        noise})$
without        the        background        Wigner        noise        term.
(See        \cite{Hisakado2026_BBP}        for        the        model        with        the        background        noise)
\end{proposition}

\begin{remark}
The        increments        $\delta_k$        are        not        independent:        since
$\delta_k=\eta_k-\eta_{k-1}$,        one        has
$\operatorname{Cov}(\delta_k,\delta_{k+1})=-\operatorname{Var}(\eta_k)=-1$
for        all        interior        $k$,        an        MA(1)-type        anti-correlation        inherited        purely
from        the        telescoping        construction,        not        from        any        assumption        on
$\{\eta_i\}$        (which        remain        i.i.d.).        This        correlation        is        essential:        it
changes        the        leading        order        of        $\mathbb        E\operatorname{Tr}(M^2)$        from
$\Theta(N^3)$        (the        value        obtained        if        the        $\delta_k$        were        treated        as
independent)        to        the        correct        $\Theta(N^2)$,        obtained        directly        from
$\mathbb        E\operatorname{Tr}(M^2)=\sum_{i,j}\mathbb        E[\eta_{\min(i,j)}^2]
=\sum_{k=1}^N(2N-2k+1)\sim        N^2$.
\end{remark}

\subsection*{B.        Numerical        findings}
\begin{observation}[Collapse        of        the        bulk]
\label{o1}
Under        the        normalisation        $S        =        M/\sqrt{N}$        that        yields        an        $O(1)$        spectrum        for        every
fixed        $\rho        <        1$,        the        top        eigenvalue        of        $S$        at        $\rho        =        1$        diverges        with        $N$        rather
than        converging,        consistent        with        the        growth        $\Theta(\sqrt{N})$        established
quantitatively        in        Fig.~\ref{conj2}        below;        correspondingly,        the        bulk        of        the        spectrum
concentrates        increasingly        near        the        origin        as        $N$        grows,        so        that        the        empirical
spectral        distribution        of        $S$        converges        weakly        to        $\delta_0$        while        a        sparse,
growing        set        of        extreme        eigenvalues        escapes        on        a        separate,        larger        scale.
\end{observation}


\begin{figure}[h]
\centering
\includegraphics[width=\columnwidth]{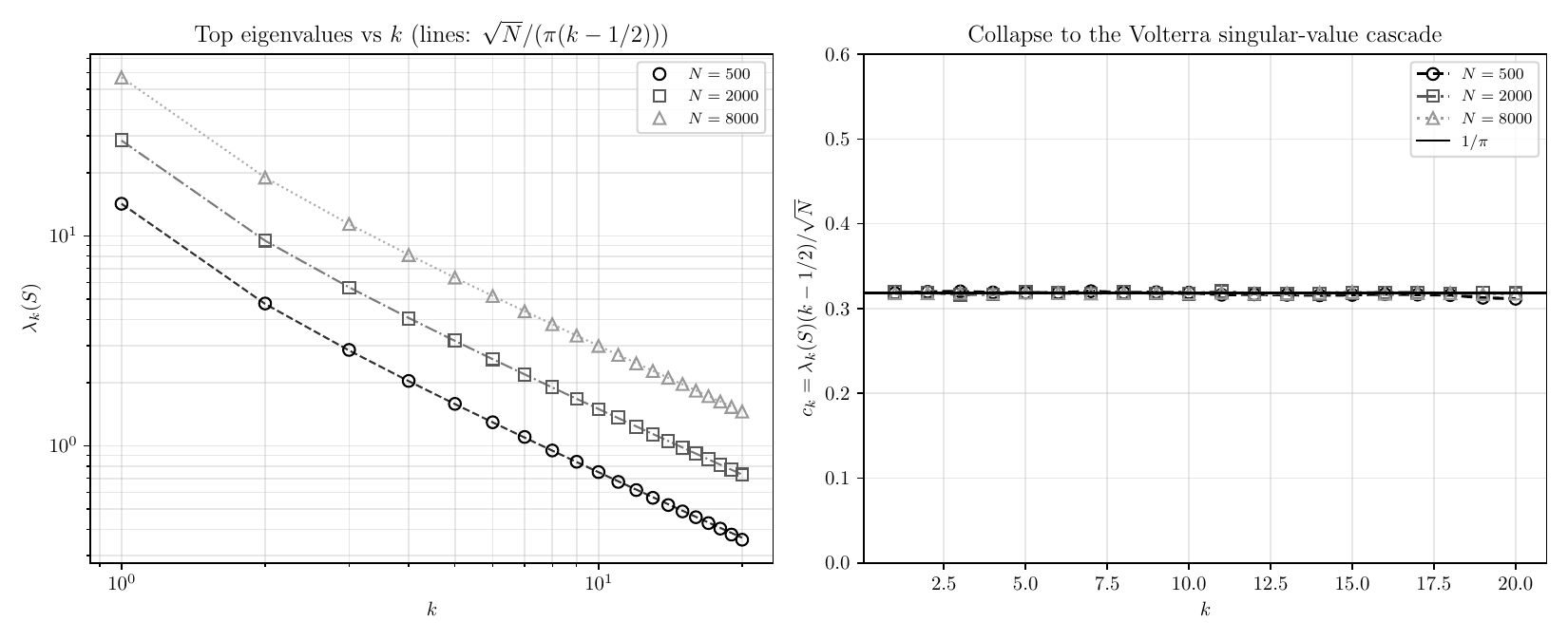}
\caption{Top        eigenvalues        $\lambda_k(S)$        of        the        noiseless        Volterra        ensemble        $S        =        M/\sqrt{N}$        (Proposition~\ref{prop:degeneration})        were        computed      exploiting        the        representation        $M=L\,\mathrm{diag}(\delta)\,L^T$,        for        $N=500,2000,8000$,        averaged        over        $40$        independent        realizations.        Left:        $\lambda_k(S)$        versus        $k$        on        a        log--log        scale        (markers),        together        with        the        predicted        curve        $\sqrt{N}/(\pi(k-1/2))$        (lines)        for        each        $N$;        a        least-squares        fit        over        $k=3,\dots,15$        gives        a        slope        of        $-1.08$        to        $-1.09$        across        the        three        values        of        $N$,        consistent        with        the        predicted        exponent        $-1$.        Right:        the        rescaled        quantity        $c_k        =        \lambda_k(S)(k-1/2)/\sqrt{N}$        versus        $k$,        collapsing        onto        the        constant        $1/\pi        \approx        0.318$        (solid        line)        for        all        three        $N$        and,        notably,        already        for        $k=1,2$        as        well,        reproducing        the        singular-value        cascade        $\sigma_k        =        1/(\pi(k-1/2))$        of        the        continuous        Volterra        integral        operator        identified        in        the        companion        paper        on        multi-critical        BBP        transitions        \cite{Hisakado2026_BBP}}
\label{conj2}
\end{figure}

\section{IV        The        Row-Independent        Power-Law        Ensemble}
\label{subsec:row-powerlaw}

\subsection{A.        Power        law        decay}
We        now        replace        the        AR(1)        row        process        of        Definition~\ref{def:row-ar1}
by        a        stationary        Gaussian        process        with        power-law        decaying        covariance,
while        retaining        exact        independence        across        rows.

\begin{definition}[Row-independent        power-law        matrix]
\label{def:row-plaw}
Fix        $N$,        $\sigma>0$,        and        $\gamma>0$.        For        each        $i=1,\dots,N$        let
$\{\eta_i(j)\}_{j=i}^N$        be        a        stationary        Gaussian        process        (generated        by
circulant        embedding        of        the        target        covariance)        with
\[
\operatorname{Cov}\big(\eta_i(j),\eta_i(j')\big)=d_{|j-j'|},
\qquad        d_t=(1+t)^{-\gamma},
\]
and        $\{\eta_i(\cdot)\}_{i=1}^N$        mutually        independent.        Define
$S_{ij}=S_{ji}:=\dfrac{\sigma}{\sqrt        N}\eta_i(j)$        for        $i\le        j$.        The
covariance        structure        across        all        entries        is        exactly
\eqref{eq:cov-structure}        with        $\rho^{|\cdot|}$        replaced        by        $d_{|\cdot|}$.
\end{definition}

\begin{lemma}[Positive-definiteness        of        the        power-law        kernel        for        every        $\gamma>0$]
\label{lem:posdef}
For        every        $\gamma>0$,        the        sequence        $d_t=(1+t)^{-\gamma}$,        $t=0,1,2,\dots$,        is        a
valid        (positive        semi-definite)        autocovariance        sequence;        equivalently,        for
every        $L$        and        every        $\gamma>0$,        the        $L\times        L$        Toeplitz        matrix
\begin{equation}
\Sigma_\gamma        \;:=\;        \big(d_{|j-k|}\big)_{j,k=1}^{L}
\end{equation}
is        positive        semi-definite.
\end{lemma}
        
\begin{proof}
By        the        Gamma-function        integral        identity,        for        $t>-1$        and        $\gamma>0$,
\begin{equation}
(1+t)^{-\gamma}        \;=\;        \frac{1}{\Gamma(\gamma)}\int_0^\infty
\lambda^{\gamma-1}        e^{-\lambda}\,        e^{-\lambda        t}\,        d\lambda,
\label{eq:gamma-mixture}
\end{equation}
using        $\int_0^\infty        \lambda^{\gamma-1}e^{-a\lambda}\,d\lambda
=\Gamma(\gamma)a^{-\gamma}$        for        $a>0$,        with        $a=1+t$.        Thus        $d_t$        is        a        mixture,
with        respect        to        the        (nonnegative)        $\mathrm{Gamma}(\gamma,1)$        density
$\lambda^{\gamma-1}e^{-\lambda}/\Gamma(\gamma)\,d\lambda$,        of        the        exponential
kernels        $t\mapsto        e^{-\lambda        t}$.        Each        exponential        kernel        $e^{-\lambda        t}$        is
the        autocovariance        sequence        of        a        stationary        AR(1)/Ornstein--Uhlenbeck-type
process        and        is        therefore        positive        semi-definite        for        every        $\lambda>0$:        for
any        $L$        and        any        $x\in\mathbb{R}^L$,
\begin{equation}
\sum_{j,k=1}^{L}        x_j        x_k\,        e^{-\lambda|j-k|}        \;\ge\;        0.
\end{equation}
Integrating        this        nonnegative        quantity        against        the        nonnegative        mixing        density
in        Eq.        \eqref{eq:gamma-mixture}        preserves        the        inequality:
\begin{equation}
\sum_{j,k=1}^{L}        x_j        x_k\,        d_{|j-k|}
\;=\;        \frac{1}{\Gamma(\gamma)}\int_0^\infty        \lambda^{\gamma-1}e^{-\lambda}
\left(\sum_{j,k=1}^{L}        x_j        x_k\,        e^{-\lambda|j-k|}\right)        d\lambda        \;\ge\;        0.
\end{equation}
Hence        $\Sigma_\gamma$        is        positive        semi-definite        for        every        $\gamma>0$.
\end{proof}
        
\begin{proposition}[Two        distinct        thresholds]
\label{prop:two-thresholds}
Let        $f_{\gamma}(\theta)=d_0+2\sum_{t\ge1}d_t\cos(\theta        t)$        be        the        spectral
symbol        of        the        row        covariance.        Then:
\begin{enumerate}
\item[(i)]        \textbf{(Bulk        moments.)}        The        fourth-moment        ``hub''        correction
of        Proposition~\ref{prop:mu4}        is        controlled        by        $\sum_t        d_t^2$,        which
converges        iff        $\gamma>\tfrac12$.        For        $\gamma>\tfrac12$        the        limiting
moments        $\mu_{2k}$        exist        and        are        finite.
\item[(ii)]        \textbf{(Flatness        /        edge.)}        $f_{\gamma}(\theta)$        is        bounded        iff
$\{d_t\}$        is        absolutely        summable,        i.e.\        iff        $\gamma>1$;        for
$\tfrac12<\gamma\le1$,        Tauberian        asymptotics        give
$f_{\gamma}(\theta)\sim        C|\theta|^{\gamma-1}\to\infty$        as        $\theta\to0$,        so
Assumption        (E)        fails.
\end{enumerate}
Consequently        $\gamma=\tfrac12$        and        $\gamma=1$        are        \emph{a        priori}
distinct        thresholds,        governing        the        bulk        and        the        edge        respectively.
\end{proposition}
\begin{proof}
The        proof                of        (i)        is        in        \cite{Hisakado2026}.

Since        $d_t        =        (1+t)^{-\gamma}        \sim        t^{-\gamma}$        as        $t        \to        \infty$,        and        a        finite-lag
shift        does        not        affect        the        small-$\theta$        singularity        structure        (Remark~11        below),
it        suffices        to        analyze
\[
                g(\theta)        :=        \sum_{t        \ge        1}        t^{-\gamma}        e^{i\theta        t}        =        \mathrm{Li}_\gamma\!\left(e^{i\theta}\right),
\]
where        $\mathrm{Li}_\gamma$        denotes        the        polylogarithm        function,        and
$f_{\gamma}(\theta)        -        d_0        \sim        2\,\mathrm{Re}\,        g(\theta)$.

By        the        classical        expansion        of        the        polylogarithm        near        $z        =        1$        on        the        unit        circle
(see        Ch.V        of        \cite{Zyg1959}),        for        non-integer        $s        >        0$,
\[
                \mathrm{Li}_s(z)        =        \Gamma(1-s)(-\log        z)^{s-1}
                                +        \sum_{k=0}^{\infty}        \frac{\zeta(s-k)}{k!}(\log        z)^k,        \qquad        z        \to        1.
\]
Setting        $z        =        e^{i\theta}$,        so        that        $\log        z        =        i\theta$        as        $\theta        \to        0^+$,        the
leading        singular        term        is
\[
                \mathrm{Li}_\gamma\!\left(e^{i\theta}\right)
                                \;\sim\;        \Gamma(1-\gamma)\,(-i\theta)^{\gamma        -        1},        \qquad        \theta        \to        0^+.
\]

Using        $(-i)^{\gamma-1}        =        e^{-i\pi(\gamma-1)/2}$        and
$\cos\!\big(\tfrac{\pi(\gamma-1)}{2}\big)        =        \sin\!\big(\tfrac{\pi\gamma}{2}\big)$,
the        real        part        is
\[
                \operatorname{Re}\!\left[\Gamma(1-\gamma)(-i\theta)^{\gamma-1}\right]
                                =        \Gamma(1-\gamma)\,\sin\!\left(\frac{\pi\gamma}{2}\right)\theta^{\gamma-1}
                                =:        C_f(\gamma)\,\theta^{\gamma-1}.
\]
For        $0        <        \gamma        <        1$        both        $\Gamma(1-\gamma)        >        0$        and        $\sin(\pi\gamma/2)        >        0$,        so
$C_f(\gamma)        >        0$        and        hence
\[
                f_{\gamma}(\theta)        \;\sim\;        2\,C_f(\gamma)\,|\theta|^{\gamma        -        1}
                                \;\longrightarrow\;        \infty        \qquad        \text{as        }        \theta        \to        0,
\]
which        is        the        claimed        divergence.        Conversely,        for        $\gamma        >        1$        the        exponent
$\gamma        -        1$        is        positive,        so        the        singular        term        vanishes        as        $\theta        \to        0$        and
$\mathrm{Li}_\gamma(1)        =        \zeta(\gamma)        <        \infty$,        recovering        the        boundedness        of
Step~1.        The        boundary        case        $\gamma        =        1$        is        treated        in        Remark~\ref{r77}        below.
\end{proof}

\begin{remark}[Shift        invariance]
\label{r8}
Since
$\sum_{t        \ge        1}        (1+t)^{-\gamma}        e^{i\theta        t}
                =        e^{-i\theta}        \sum_{t'        \ge        2}        (t')^{-\gamma}        e^{i\theta        t'}$,
the        difference        from        $g(\theta)$        consists        of        removing        finitely        many        terms
(here        $t'=1$)        and        multiplying        by        a        smooth,        bounded        factor        $e^{-i\theta}$;
neither        operation        alters        the        order        of        the        singularity        as        $\theta        \to        0$.
\end{remark}

\begin{remark}[Boundary        case        $\gamma        =        1$]
\label{r77}
At        $\gamma        =        1$        the        expansion        above        is        not        directly        applicable        ($s=1$        is        the
excluded        logarithmic        case),        and        instead
$\mathrm{Li}_1(e^{i\theta})        =        -\log(1        -        e^{i\theta})        \sim        -\log(i\theta)$,
so        $f_{\gamma}(\theta)$        diverges        logarithmically,        $f_{\gamma}(\theta)        \sim        -2\log|\theta|$,
rather        than        as        a        power        law;        the        conclusion        of        unboundedness        at        $\gamma=1$        is
unaffected.
\end{remark}
\subsection{B.        Explicit        Bulk        Self-Consistent        Equation        for        the        Power-Law        Ensemble}

The        scalar        self-consistent        equation,  Eq.  \eqref{eq:sce}        of        Proposition~\ref{p4}        for        the        exponential-decay        ensemble        depends        on
the        row        covariance        only        through        its        power        spectral        density,        and
therefore        generalizes        verbatim        to        the        power-law        ensemble        of
Definition~\ref{def:row-plaw}.

\begin{proposition}(Scalar        self-consistent        equation,        power-law
case).        Let
\begin{equation}
f_\gamma(\theta)        \;:=\;        \sum_{t=-\infty}^{\infty}        d_{|t|}\,e^{i\theta        t},
\qquad        d_t        =        (1+t)^{-\gamma},
\label{eq:ftheta_gamma}
\end{equation}
the        power        spectral        density        of        the        row        covariance        of
Definition~\ref{def:row-plaw}.       
Then        the        Stieltjes        transform        of        the        limiting        spectral
measure        of        $S$        satisfies
\begin{equation}
m(z)        \;=\;        \int_0^{2\pi}\frac{d\theta}{2\pi}\,
\frac{1}{-z-f_\gamma(\theta)\,m(z)},
\qquad        z\in\mathbb{C}^+,
\label{eq:sce_gamma}
\end{equation}
identical        in        form        to        Eq.~\eqref{eq:ftheta}.
        Unlike        the        AR(1)        case,        $f_\gamma$        has
no        closed        elementary        form;        for        $1/2<\gamma\le1$        it        is        unbounded
as        $\theta\to0$        (Proposition~\ref{prop:two-thresholds}(ii)),        and        Eq.        \eqref{eq:sce_gamma}
remains        well        defined        despite        this        because        the        singularity        is
integrable        and        the        integrand        vanishes        there.
\end{proposition}

\begin{proof}
Identical        to        the        proof        of        Proposition~\ref{p4},;        the        argument        depends        only        on        the        row
covariance        being        row-independent        and        Toeplitz        (Definition~\ref{def:row-plaw}),        not        on        the
specific        decay        mechanism.
\end{proof}

\begin{remark}        (A        spectral        proof        of        Proposition~\ref{prop:two-thresholds}(i)).        
\label{r6}
Equation        \eqref{eq:ftheta_gamma}        identifies        $\{d_t\}$        as        the        Fourier
coefficients        of        $f_\gamma$,        so        Parseval's        identity        gives
\begin{equation}
\int_0^{2\pi}\frac{d\theta}{2\pi}\,f_\gamma(\theta)^2
\;=\;        \sum_{t=-\infty}^{\infty}        d_t^2
\;=\;        d_0^2+2\sum_{t\ge1}(1+t)^{-2\gamma}.
\label{eq:parseval}
\end{equation}
The        sum        on        the        right        converges        if        and        only        if        $\gamma>1/2$.        Since
the        hub        correction        to        $\mu_4$        in        Proposition~\ref{prop:mu4} (and        its
power-law        analogue        underlying        Proposition~\ref{prop:two-thresholds}(i))        is        controlled
by        exactly        this        quantity,        $\sum_t        d_t^2$,        Eq.        
\eqref{eq:parseval}        gives        an        independent,        spectral        derivation        of
the        threshold        $\gamma_c=1/2$        of        Proposition~\ref{prop:two-thresholds}(i),        complementing
the        direct        combinatorial        argument        of        \cite{Hisakado2026}:        the        bulk
fourth-moment        is        finite        iff        $f_\gamma\in        L^2[0,2\pi]$.
\end{remark}

Figure~\ref{fig:bulk_sce_gamma}        shows        the        corresponding        comparison.        For        $\gamma=0.3$,
where        Proposition~\ref{prop:two-thresholds}(i)        already        establishes        that        the        bulk        moments
diverge        as        $N\to\infty$,        we        plot        the        density        on        a        logarithmic
scale        to        display        the        full        extent        of        the        heavy        tail        together        with
the        bulk;        the        self-consistent        theory        tracks        the        empirical
histogram        closely        across        several        orders        of        magnitude        in        density,
including        the        tail        region        beyond        $|x|\gtrsim5$,        whereas        the
semicircle        law        is        confined        to        $|x|\le2$        by        construction        and
provides        no        description        of        this        region        at        all.

\begin{figure}[h]
\centering
\includegraphics[width=\textwidth]{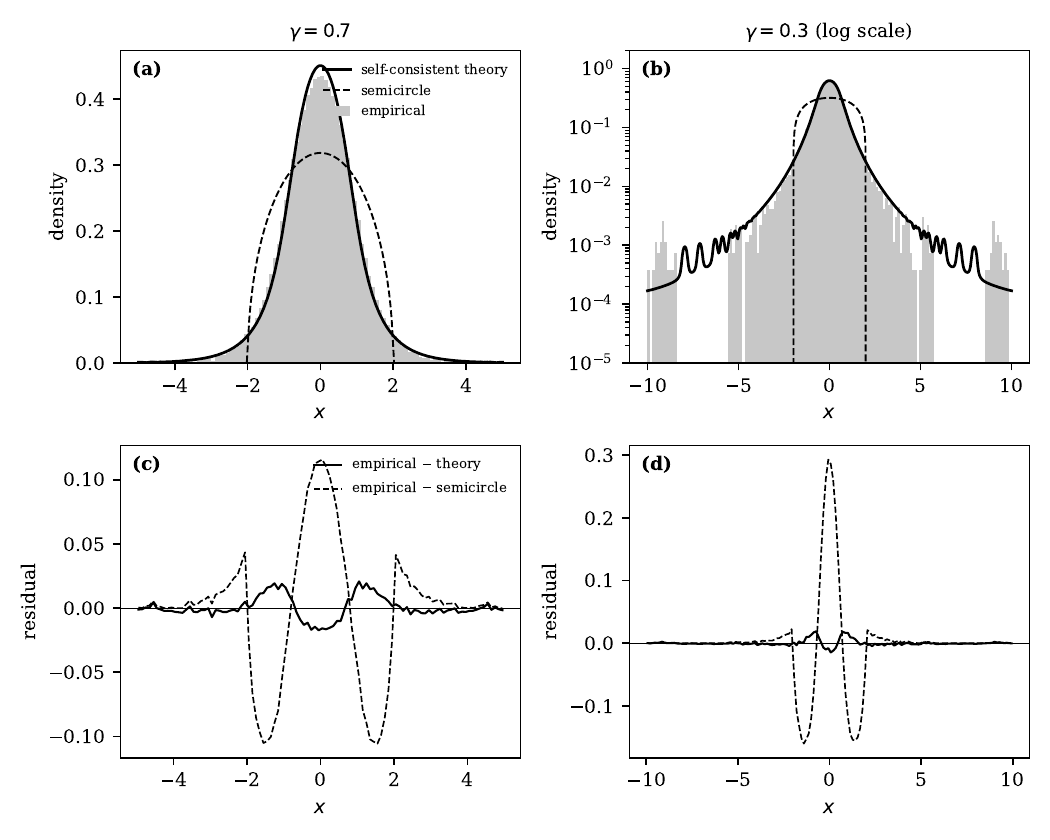}
\caption{(a,b)        Bulk        spectral        density        of        the        power-law        ensemble
(Definition~2)        at        $\gamma=0.7$        (linear        scale)        and        $\gamma=0.3$
(logarithmic        scale,        showing        the        heavy        tail):        empirical        histogram
(gray),        self-consistent        theory        \eqref{eq:sce_gamma}        (solid
black),        and        the        semicircle        law        (dashed        black).        (c,d)        Residuals
of        the        theory        and        of        the        semicircle        law        relative        to        the
empirical        histogram.        $N=800$,        $25$        realizations,
$\eta=0.03$        ($\gamma=0.7$)        or        $0.05$        ($\gamma=0.3$).}
\label{fig:bulk_sce_gamma}
\end{figure}

As        in        Section~II.D,        we        regard        Eq.~\eqref{eq:sce_gamma}
with        
Figure~\ref{fig:bulk_sce_gamma}        as        a        quantitative        confirmation
of        the        qualitative        statement        "Bulk        statistics:        Not        semi-circlar"
recorded        for        every        regime        in        Remark~\ref{r6}        as
supplying        an        independent,        spectral        route        to        the        threshold
$\gamma_c=1/2$        that        is        derived,        in        this        section,        directly        from
the        same        equation        that        also        reproduces        the        full        bulk        density.
\subsection{C        Phase        transition        of        the        fourth-moment:        exact        verification
and        critical        exponents}

Remark~\ref{r6}        identified        $\gamma_c=1/2$        as        the        threshold        for        divergence        of
the        bulk        fourth-moment,        via        the        spectral        criterion        $f_\gamma\in
L^2[0,2\pi]$,        equivalently        $\sum_t        d_t^2<\infty$.        Here        we        make        this
threshold        fully        rigorous        ---        not        merely        as        an        asymptotic        $N\to\infty$
statement        inferred        from        the        bulk        moment        calculation,        but        as        an        exact,
finite-$N$,        combinatorial        identity        whose        $N\to\infty$        limit        exhibits        a
genuine        non-analyticity        at        $\gamma=1/2$,        with        critical        exponents        that
we        compute        in        closed        form.

\subsubsection{i.        Exact        finite-$N$        formula}

\begin{proposition}[Exact        finite-$N$        fourth-moment]
Let        $S$        be        the        row-independent        ensemble        of        Definition~\ref{def:row-plaw}      with        $\sigma=1$,
and        write        $R(a,b):=\min(a,b)$,        $T(a,b):=\max(a,b)$,        so        that
$\mathrm{Cov}(S_{ab},S_{cd})        =        \tfrac1N\,d_{|T(a,b)-T(c,d)|}\,
\mathbf        1[R(a,b)=R(c,d)]$.        Then,        for        every        finite        $N$,
\begin{equation}
\mathbb        E\,\mathrm{Tr}(S^4)        =        2\sum_{a,b=1}^N        f(a,b)^2        \;+\;\Xi_N,
\qquad
f(a,b):=\sum_{c=1}^N        \mathrm{Cov}(S_{ac},S_{cb}),
\label{eq:exactTr4-final}
\end{equation}
where        $\Xi_N:=\sum_{i,j,k,l=1}^N
\mathrm{Cov}(S_{ij},S_{kl})\,\mathrm{Cov}(S_{jk},S_{li})$,        and        we        have
used        the          symmetry        $f(a,b)=f(b,a)$.
\end{proposition}
\begin{proof}
By        Wick's        theorem,        $\mathbb        E[S_{ij}S_{jk}S_{kl}S_{li}]$        is        the        sum        of
three        pairings;        the        pairings        $(ij,jk)(kl,li)$        and        $(ij,li)(jk,kl)$        each
factor,        after        relabeling        and        using
$\mathrm{Cov}(S_{ab},S_{cd})=\mathrm{Cov}(S_{ba},S_{cd})=\mathrm{Cov}(S_{ab},S_{dc})$,
into        $\sum_{a,b}f(a,b)^2$;        the        pairing        $(ij,kl)(jk,li)$        does        not        factor
and        is        exactly        $\Xi_N$.
\end{proof}

\subsubsection{ii.        The        $N\to\infty$        limit        is        a        genuine        phase        transition}

\begin{proposition}[Non-analyticity        at        $\gamma_c=1/2$        and                the        susceptibility]
\label{p14}
Define        $\mu_4(\infty,\gamma):=\lim_{N\to\infty}\mathbb        E\,\mathrm{Tr}(S^4)/N$.
Then
\begin{equation}
\mu_4(\infty,\gamma)        =        2+\frac43\bigl(\zeta(2\gamma)-1\bigr)
\quad\text{for        }\gamma>\tfrac12,
\qquad
\mu_4(\infty,\gamma)=+\infty
\quad\text{for        }\gamma\le\tfrac12,
\label{eq:mu4inf}
\end{equation}
where        $\zeta$        is        the        Riemann        zeta        function.        In        particular
$\mu_4(\infty,\cdot)$        is        real-analytic        on        $(1/2,\infty)$        and        identically
infinite        on        $(0,1/2]$:        the        bulk        transition        at        $\gamma_c=1/2$        is        a
rigorously        established        non-analyticity        of        the        infinite-volume        limit,
not        merely        a        finite-size        numerical        observation.
\end{proposition}
\begin{proof}
By        the        factored        form        $2\sum_{a,b}f(a,b)^2$        dominating
        $\mathbb        E\,\mathrm{Tr}(S^4)/N$        as        $N\to\infty$,        the
$N\to\infty$        limit        reduces        to        the        combinatorial        calculation        of        \cite{Hisakado2026},        giving
$\mu_4(\infty,\gamma)=2+\tfrac43\sum_{i=1}^\infty(1+i)^{-2\gamma}
=2+\tfrac43(\zeta(2\gamma)-1)$        whenever        the        sum        converges,        i.e.\
$2\gamma>1$.        For        $\gamma\le1/2$,        $\sum_i(1+i)^{-2\gamma}$        diverges        by
the        standard        comparison        test        for        $p$-series        ($2\gamma\le1$),        so
$\mu_4(\infty,\gamma)=+\infty$.
\end{proof}
\begin{corollary}[Divergence of the top eigenvalue for $\gamma \le 1/2$]
\label{cor:lmax-divergence}
Fix $\gamma \le 1/2$. Then
\[
\mathbb{E}\big[\lambda_{\max}(S_N)^4\big] \;\longrightarrow\; \infty
\qquad (N \to \infty).
\]
In particular, $\lambda_{\max}(S_N)$ is not bounded in $L^4$, uniformly in $N$.

This statement concerns the random top eigenvalue $\lambda_{\max}(S_N)$ directly; it
makes no reference to, and carries no implication for, the self-consistent edge
$\tau_0(\gamma)$ of Definition~\ref{def3}, which is in any case undefined throughout $\gamma \le 1$
since $\sigma_1(\gamma) = 2\zeta(\gamma)-1 = \infty$ there (Lemma~\ref{lem:sigma1-gamma}).
\end{corollary}

\begin{proof}
By the deterministic inequality $\operatorname{Tr}(S_N^4) = \sum_{i=1}^N \lambda_i(S_N)^4
\le N \, \lambda_{\max}(S_N)^4$, valid pointwise for every realization of $S_N$, we have
\[
\lambda_{\max}(S_N)^4 \;\ge\; \frac{\operatorname{Tr}(S_N^4)}{N}.
\]
Taking expectations preserves the inequality:
\[
\mathbb{E}\big[\lambda_{\max}(S_N)^4\big] \;\ge\; \frac{\mathbb{E}[\operatorname{Tr}(S_N^4)]}{N}.
\]
By Proposition~14, $\mu_4(\infty,\gamma) := \lim_{N\to\infty} \mathbb{E}[\operatorname{Tr}(S_N^4)]/N
= +\infty$ for every $\gamma \le 1/2$. Hence the right-hand side diverges, and so does
the left-hand side.
\end{proof}

\begin{corollary}[Exact    limit    of    the    dominant    sub-case,    and    boundedness    of    $\Xi_N$    for    $\gamma>1/2$]
\label{c13}
Fix    $\gamma    >    1/2$.    Define    the    partial    sum
\[
D_2(m)    :=    \sum_{t=1}^m    d_t^2,    \qquad    d_t    =    (1+t)^{-\gamma},
\]
so    that    $D_2(m)$    increases    monotonically    to    $D_2(\infty)    =    \zeta(2\gamma)    -    1    <    \infty$.    The    sub-case    of
$\Xi_N$    obtained    by    restricting    its    defining    sum    to    $k    =    i$,
\[
\Xi_N^{(i=k)}    :=    \sum_{i,j,l=1}^N    \mathrm{Cov}(S_{ij},    S_{il})^2,
\]
admits    the    exact    closed    form
\[
\Xi_N^{(i=k)}    \;=\;    1    +    \frac{2}{N^2}\sum_{L=1}^N    \sum_{s=1}^{L-1}    (L-s)\,    d_s^2,
\]
    converges,    as    $N    \to    \infty$,    to
\[
\Xi_N^{(i=k)}    \;\longrightarrow\;    1    +    D_2(\infty)    \;=\;    \zeta(2\gamma).
\]
and
\[
\Xi_N    =    O(1),    \qquad    \text{in    particular    }    \Xi_N/N    \to    0.
\]
This    is    precisely    the    input    used    in    the    proof    of
Proposition~\ref{p14}:    dividing    Eq~\eqref{eq:exactTr4-final}  by    $N$,
\[
\mu_4(\infty,\gamma)    =    \lim_{N\to\infty}    \frac{2\sum_{a,b}    f(a,b)^2}{N}
+    \lim_{N\to\infty}    \frac{\Xi_N}{N}    =    2    +    \frac{4}{3}\bigl(\zeta(2\gamma)    -    1\bigr)    +    0,
\]
matching    Eq.~(\ref{eq:mu4inf}).
The    numerical    confirmation  is  in  Table~\ref{tab:subleading}.
The  exact  proof    for  $\gamma>1/2$  is  in  Appendix  C.
\end{corollary}

\begin{proof}
\textbf{Exact    closed    form.}    Fix    $i$.    For    $j,    l    \ge    i$,    $\min(i,j)    =    \min(i,l)    =    i$
automatically,    so    $\mathrm{Cov}(S_{ij},    S_{il})    =    \frac{1}{N}    d_{|j-l|}$.    For    $j,    l    <    i$,
$\min(i,j)    =    j$    and    $\min(i,l)    =    l$,    so    the    covariance    is    nonzero    only    when    $j    =    l$,    in    which
case    it    equals    $\frac{1}{N}$.    In    the    mixed    case    (exactly    one    of    $j,    l$    less    than    $i$)    the
covariance    vanishes.    Writing    $L    :=    N    -    i    +    1$    for    $|\{i,\dots,N\}|$,
\[
\sum_{j,l=i}^N    d_{|j-l|}^2    =    L    +    2\sum_{s=1}^{L-1}    (L-s)\,    d_s^2,
\]
while    the    diagonal    $j    =    l    <    i$    contributes    $i    -    1$    further    unit    terms.    Since    $L    +    (i-1)    =    N$,
summing    over    $i    =    1,    \dots,    N$    and    substituting    $L    =    N    -    i    +    1$    gives    the    stated    closed    form.

\textbf{Limit.}    Reversing    the    order    of    summation,
\[
\sum_{L=1}^N    \sum_{s=1}^{L-1}    (L-s)\,    d_s^2
=    \sum_{s=1}^{N-1}    d_s^2    \sum_{L=s+1}^N    (L-s)
=    \frac{1}{2}\sum_{s=1}^{N-1}    d_s^2\,    (N-s)(N-s+1),
\]
so    that
\[
\frac{2}{N^2}\sum_{L=1}^N    \sum_{s=1}^{L-1}    (L-s)\,    d_s^2
=    \frac{1}{N^2}\sum_{s=1}^{N-1}    d_s^2\,    (N-s)(N-s+1).
\]
For    each    fixed    $s$,    $(N-s)(N-s+1)/N^2    \to    1$    as    $N    \to    \infty$,    and    this    ratio    is    bounded    by    $2$
uniformly    in    $1    \le    s    \le    N-1$;    since    $\sum_s    d_s^2    <    \infty$    for    $\gamma    >    1/2$,    dominated
convergence    gives
\[
\frac{1}{N^2}\sum_{s=1}^{N-1}    d_s^2\,    (N-s)(N-s+1)    \longrightarrow    \sum_{s=1}^\infty    d_s^2    =    D_2(\infty),
\]
hence    $\Xi_N^{(i=k)}    \to    1    +    D_2(\infty)    =    \zeta(2\gamma)$.

\textbf{Cyclic    symmetry.}    The    summand    of    $\Xi_N$,    $\mathrm{Cov}(S_{ij},S_{kl})\mathrm{Cov}(S_{jk},S_{li})$,
is    invariant    under    the    simultaneous    relabeling    $i    \to    j    \to    k    \to    l    \to    i$,    which    sends    it    to
$\mathrm{Cov}(S_{jk},S_{li})\mathrm{Cov}(S_{kl},S_{ij})$    —    the    same    product    with    its    two    factors
exchanged.    Under    this    relabeling    the    condition    $k=i$    becomes    $l=j$,    so    summing    over    all
$(i,j,k,l)$    gives    $\Xi_N^{(l=j)}    =    \Xi_N^{(i=k)}$    exactly.
\end{proof}
\subsubsection{iii.        Critical        exponents}

\begin{proposition}[Super-critical        exponent:        simple        pole,        $\psi=1$]
\label{prop:super-critical-exponent}
As        $\gamma\to\tfrac12^+$,
\begin{equation}
\mu_4(\infty,\gamma)        \sim        \frac{2}{3}\cdot\frac{1}{\gamma-1/2}.
\label{eq:psi1}
\end{equation}
\end{proposition}
\begin{proof}
$\zeta(s)$        has        a        simple        pole        at        $s=1$        with        residue        $1$:
$\zeta(1+2\varepsilon)=\tfrac1{2\varepsilon}+\gamma_E+O(\varepsilon)$
($\gamma_E$        the        Euler--Mascheroni        constant).        Substituting
$2\varepsilon=2\gamma-1$        into        Eq.~\eqref{eq:mu4inf}        gives        the        stated
divergence.       
\end{proof}

\begin{proposition}[Sub-critical        exponent:        $\theta(\gamma)=1-2\gamma$]
\label{prop15}
Fix        $\gamma<1/2$.        As        $N\to\infty$,
\begin{equation}
\mu_4(N,\gamma)-2        \;\sim\;        C_{\mu}(\gamma)\,N^{1-2\gamma},
\qquad
C_{\mu}(\gamma)        =        \frac43\,B(1-2\gamma,4)
\label{eq:subcritical}
\end{equation}
where        $B$        is        the        Euler        beta        function.        The        exponent        $\theta(\gamma)=1-2\gamma$
is        exact;        the        amplitude        $C_{\mu}(\gamma)$        is        a        leading-order        estimate        (see
Remark\ref{r12}).
\end{proposition}
\begin{proof}[Proof        sketch]
We        show        that        $\Xi_N$        is        subleading        relative        to        $\sum_{a,b}f(a,b)^2$        as
$N\to\infty$        for        fixed        $\gamma<1/2$,        so        that        the        leading        behavior        of
$\mu_4(N,\gamma)-2$        is        governed        entirely        by        the        factored        term.

\emph{Step        1        (leading        behavior        of        $f$).}        For        $a<b$,        the        dominant
contribution        to        $f(a,b)=\sum_c\mathrm{Cov}(S_{ac},S_{cb})$        comes        from
$c<a$,        where        $R(a,c)=R(c,b)=c$        automatically        and
$T(a,c)=a,\,T(c,b)=b$,        giving        $f(a,b)\approx        (a/N)\,d_{b-a}$        (plus
$O(1/N)$        boundary        corrections        from        $c\in[a,b)$).

\emph{Step        2        (leading        behavior        of        $\sum        f^2$).}        Substituting        into
$\sum_{a,b}f(a,b)^2\approx2\sum_{s=1}^{N-1}d_s^2\sum_{a=1}^{N-s}(a/N)^2$
and        passing        to        the        continuum        limit        ($s=Nu$)        gives
$\sum_{a,b}f(a,b)^2\sim\tfrac23B(1-2\gamma,4)\,N^{2-2\gamma}$,        using
$d_s^2\sim        s^{-2\gamma}$        and        the        Beta-function        integral
$\int_0^1u^{-2\gamma}(1-u)^3\,du=B(1-2\gamma,4)$.

\emph{Step        3        ($\Xi_N$        is        subleading).}        The        dominant        sub-case        of        $\Xi_N$
with        $k=i$        (equivalently,        by        symmetry,        $l=j$)        contributes
$\Xi_N^{(i=k)}=\tfrac{1}{N^2}\sum_{m=1}^N        m\,D_2(m)$,        where        $D_2(m)$        is        in        Corollary        \ref{c13}.        For
$\gamma<1/2$,        $D_2(m)\sim        c'm^{1-2\gamma}$,        giving
$\Xi_N^{(i=k)}\sim        N^{1-2\gamma}$        ---        one        power        of        $N$        below
$\sum_{a,b}f(a,b)^2\sim        N^{2-2\gamma}$.    
Since  $\Xi_N=o(N)=o(\mathrm{Term\  C})$  (the  proof  is  Corollary  \ref{c44}  in  Appendix  C.)
  $\Xi_N$        does        not        affect        the        leading        exponent.
\end{proof}

\begin{table}[h]
\centering
\begin{tabular}{c|ccc|ccc|cc}
\hline\hline
&  \multicolumn{3}{c|}{$N=50$}  &  \multicolumn{3}{c|}{$N=200$}  &  \multicolumn{2}{c}{$\Xi_N^{(i=k)}/\Xi_N$}  \\
$\gamma$  &  Term  C  &  $\Xi_N$  &  $\Xi_N/\text{Term  C}$  &  Term  C  &  $\Xi_N$  &  $\Xi_N/\text{Term  C}$  &  $N{=}50$  &  $N{=}200$  \\
\hline
0.3  &  352.90  &  10.113  &  0.0287  &  2589.91  &  20.257  &  0.0078  &  0.549  &  0.525  \\
0.4  &  278.55  &  7.251    &  0.0260  &  1721.07  &  12.210  &  0.0071  &  0.569  &  0.541  \\
0.7  &  171.86  &  3.358    &  0.0195  &  786.58    &  4.051    &  0.0052  &  0.649  &  0.623  \\
1.0  &  134.05  &  2.077    &  0.0155  &  558.05    &  2.211    &  0.0040  &  0.741  &  0.726  \\
1.5  &  112.31  &  1.380    &  0.0123  &  452.59    &  1.398    &  0.0031  &  0.862  &  0.858  \\
\hline\hline
\end{tabular}
\caption{Exact  finite-$N$  values  of  Term  C  $=2\sum_{a,b}f(a,b)^2$,  $\Xi_N$,
and  their  ratio,  at  the  endpoints  $N=50$  and  $N=200$  of  the  range  tested
(intermediate  $N=100,150$  interpolate  monotonically  and  are  omitted  for
brevity),  across  the  full  range  $\gamma\in\{0.3,0.4,0.7,1.0,1.5\}$  spanning
the  sub-critical  ($\gamma<1/2$)  and  super-critical  ($\gamma>1/2$)  regimes.
The  ratio  $\Xi_N/\text{Term  C}\to0$  in  every  case;  this  is  no  longer  merely
numerical,  but  is  established  rigorously  for  every  $\gamma>0$  by  Corollary
\ref{c44}  (Appendix  C.6),  which  additionally  gives  $\Xi_N=O(1)$  outright  for
$\gamma>1/2$  (Corollary  \ref{lc24}).  The  final  two  columns  show  that  the  dominant
sub-case  $\Xi_N^{(i=k)}$  (Corollary  \ref{c13})  accounts  for  an  increasing  share  of
$\Xi_N$  as  $\gamma$  grows:  it  trends  toward  $1/2$  in  the  sub-critical
regime,  consistent  with  the  conjecture  that  the  $i=k$  and  $l=j$  sub-cases
together  exhaust  $\Xi_N$  there,  but  stabilizes  well  above  $1/2$  (up  to
$\approx0.86$  at  $\gamma=1.5$)  once  $\gamma>1/2$  —  indicating  overlap
between  the  two  sub-cases  in  the  super-critical  regime  that  we  have  not
evaluated  in  closed  form,  and  which  remains  open.}
\label{tab:subleading}
\end{table}

\begin{remark}[Status        of        the        amplitude        $C_{\mu}(\gamma)$]
The        exponent        $\theta(\gamma)=1-2\gamma$        in
Eq.~\eqref{eq:subcritical}        rests        on        Step~3        (numerically        confirmed,
Table~\ref{tab:subleading})        and        the        continuum        approximation        of        Step~2.
The        predicted        amplitude        $C_{\mu}(\gamma)=\tfrac43B(1-2\gamma,4)$        agrees        with
the        exact        Term~C        only        to        within        $11$--$15\%$        at        $\gamma=0.3$        and
$37$--$42\%$        at        $\gamma=0.4$        for        $N\le200$        (the        discrepancy        shrinking
slowly        as        $N$        grows),        reflecting        $O(1/N)$        boundary        corrections        to        the
Step~1        approximation        of        $f(a,b)$        that        we        have        not        carried        to
next-to-leading        order.        We        therefore        regard        the        exponent
$\theta(\gamma)=1-2\gamma$        as        established,        while        the        precise        amplitude        $C_{\mu}(\gamma)$        remains        a
leading-order        estimate        rather        than        an        exact        result.
\label{r12}
\end{remark}

\begin{remark}[Continuously-varying        critical        exponents]
The        exponent        $\theta(\gamma)=1-2\gamma$        is        not        a        single        universal
number        but        an        explicit        function        of        the        correlation-decay        parameter
$\gamma$        itself,        varying        continuously        throughout        the        sub-critical
region        $\gamma<1/2$.        This        is        not        the        exponent        structure        of        a
short-range        transition,        but        is
characteristic        of        long-range-interacting        statistical        systems,        where
interactions        decaying        as        $r^{-(d+\sigma)}$        produce        critical        exponents
depending        continuously        on        $\sigma$        below        an        upper        critical        value
\cite{FisherMaNickel1972}.        Here        $\gamma$        plays        the        role        of        $\sigma$:
the        temporal        correlation        decay        exponent        directly        and        continuously
tunes        the        anomalous        dimension        of        the        bulk        susceptibility,
\label{r13}
\end{remark}

\begin{observation}[Numerical        confirmation,        full        $\gamma$-range]
Table~\ref{tab:mu4exact-final}        reports        $\mu_4^{(N)}$,        computed        exactly
via        Eq.~\eqref{eq:exactTr4-final}        (no        simulation),        for        $N$        up        to        $200$
and        $\gamma$        spanning        $\gamma_c=1/2$;        the        fitted        exponent
$s_4:=\log(\mu_4^{(200)}/\mu_4^{(80)})/\log(200/80)$        decreases        smoothly
and        monotonically        across        $\gamma_c$,        consistent        with
$\theta(\gamma)=1-2\gamma$.
\end{observation}

\begin{table}[h]
\caption{Exact        fourth        moment        $\mu_4^{(N)}$        and        fitted        exponent        $s_4$,
compared        with        $1-2\gamma$.}
\label{tab:mu4exact-final}
\begin{tabular}{cccccc}
\hline\hline
$\gamma$        &        $\mu_4^{(80)}$        &        $\mu_4^{(200)}$        &        $s_4$        &        $1-2\gamma$\\
\hline
0.30        &        8.824        &        13.051        &        0.427        &        0.400\\
0.40        &        6.584        &        8.666                &        0.300        &        0.200\\
0.45        &        5.802        &        7.277                &        0.247        &        0.100\\
0.50        &        5.178        &        6.228                &        0.202        &        0        (log)\\
0.55        &        4.676        &        5.429                &        0.163        &        0\\
0.60        &        4.270        &        4.812                &        0.131        &        0\\
0.70        &        3.666        &        3.953                &        0.082        &        0\\
\hline\hline
\end{tabular}
\end{table}

\subsection{D.        Theoretical        threshold        from        the        breakdown        of        the        flatness        condition}
\label{sec:IVB}
        
By        Proposition~\ref{prop:two-thresholds},        two        distinct        thresholds        govern        the
row-independent        power-law        ensemble        of        Definition~\ref{def:row-plaw}:
$\gamma_c        =        1/2$,        at        which        the        fourth-moment        $\mu_4$        diverges        (an
$\ell^2$-type        condition        on        the        correlation        sequence,        $\sum_t        d_t^2        <\infty$),
and        $\gamma_c        =        1$,        at        which        the        flatness        condition        (Assumption~(E))        required
by        the        MDE                framework        of~\cite{AEKS2020}        breaks        down        (a
strictly        stronger        $\ell^1$-type,        Wiener        condition,        $\sum_t        |d_t|        <        \infty$).
It        is        the        latter        threshold        that        is        relevant        to        edge        universality,        since        it        is
Assumption~(E),        and        not        the        finiteness        of        $\mu_4$,        that        enters        the        hypotheses
of~\cite{AEKS2020}        governing        regularity        of        the        spectral        edge.

\begin{remark}[Logical        status        of        the        Assumption~(E)        threshold]
\label{rem:logical-status}
The        breakdown        of        Assumption~(E)        at        $\gamma=1$        (Proposition~\ref{prop:two-thresholds}(ii))
is        a        rigorous,        analytic        fact.        However,        it        marks        the        failure        of        a
sufficient        condition        used        by        one        specific        proof        technique        ---        the
MDE        approach        of~\cite{AEKS2020}        ---        and        is        not,        by        itself,        a
statement        about        the        behavior        of        $\lambda_{\max}$.        In        particular:
\begin{enumerate}
\item[(i)]        For        $\gamma>1$,        Assumption~(E)        holds,        placing        the        ensemble        within
the        scope        of~\cite{AEKS2020}        modulo        Assumption~(G)        ---        exactly        the        same
logical        position        as        the        AR(1)        ensemble        of        Section~II        for        any        fixed
$\rho<1$        (Proposition~\ref{prop:assumptions}).        Assumption~(G)        (boundedness        of
the        self-consistent        equation        solution        $M(z)$        near        the        edge)        is        not
automatic        from        (A)--(E)        and        has        not        been        verified        for        either        ensemble.        We
therefore        regard        TW        edge        universality        for        $\gamma>1$        as        a
conjecture        on        the        same        footing        as        Conjecture~\ref{conj:TW}
for        the        AR(1)        case,        not        as        a        consequence        of
Proposition~\ref{prop:two-thresholds}(ii)        alone.
\item[(ii)]        For        $1/2<\gamma\le        1$,        the        failure        of        Assumption~(E)        means        only
that        this        particular        proof        route        is        unavailable.        It        does        not        constitute
evidence,        let        alone        proof,        that        $\lambda_{\max}$        behaves        differently        in        this
range.        \emph{A        priori},        edge        universality        could        persist        beyond        $\gamma=1$        via
a        different        argument,        or        could        break        down        at        some        other        threshold        entirely,
including        possibly        $\gamma_c=1/2$        itself.
\end{enumerate}
\end{remark}

\subsection{E.      
Local  Edge  Analysis  at  the  $\gamma_c=1$  Threshold
}
The  point  $\gamma_c=1$  is  special  as  the  flatness  condition  breaks  down.  We  therefore  examine  whether  this  threshold  corresponds  to  a  genuine  spectral  transition.
\begin{lemma}[Essential supremum for the power-law ensemble]
\label{lem:sigma1-gamma}
For every $\gamma > 1$, $f_\gamma(\theta)$ attains its essential
supremum at $\theta=0$, with
\[
 f_\gamma(0) = \sigma_1(\gamma) =\sum_{t=-\infty}^{\infty} d_{|t|}
= 1 + 2\sum_{t\ge1}(1+t)^{-\gamma} = 2\zeta(\gamma) - 1.
\]
\end{lemma}
\begin{proof}
Since $d_t = (1+t)^{-\gamma} > 0$ for every $t\ge0$, each term
$d_{|t|}\cos(\theta t)$ of $f_\gamma(\theta) = d_0 + 2\sum_{t\ge1} d_t
\cos(\theta t)$ is maximized simultaneously at $\theta=0$, where
$\cos(\theta t) = 1$ for every $t$. Hence
$f_\gamma(\theta) \le f_\gamma(0)$ for all $\theta$, with equality
only at $\theta=0$, so $\operatorname{ess\,sup} f_\gamma = f_\gamma(0)$.
Evaluating the sum at $\theta=0$,
\[
f_\gamma(0) =\sigma_1(\gamma) = d_0 + 2\sum_{t\ge1} d_t
= 1 + 2\sum_{t\ge1} (1+t)^{-\gamma}
= 2\zeta(\gamma) - 1,
\]
which converges for $\gamma > 1$ by the standard $p$-series test.
\end{proof}
\begin{definition}
\label{def:Dgamma}
For $\gamma > 0$, $\gamma \neq 1$, define
\begin{equation}
D(\gamma) := -2\,\Gamma(1-\gamma)\,\cos\!\left(\frac{\pi(\gamma-1)}{2}\right).
\end{equation}
\end{definition}

\begin{remark}
By Proposition \ref{prop:two-thresholds}'s derivation, $D(\gamma)$
coincides, 
with the leading singular
coefficient $C_f(\gamma)$ governing $f_{\sigma}(\theta) \sim 2\,C_f(\gamma)\,|\theta|^{\gamma-1}$
as $\theta \to 0^+$; equivalently, vi, $D(\gamma) = -(2\pi)^{1-\gamma}\zeta(\gamma)/\zeta(1-\gamma)$.
\end{remark}

\begin{proposition}[Exact        cancellation        of        the        pole        structure]
Let        $\epsilon:=\gamma-1$.        As        $\gamma\to1^+$,
\[
D(\gamma)        -        \sigma_1(\gamma)        \;=\;        1        +        O(\epsilon).
\]
In        particular        $D(\gamma)-\sigma_1(\gamma)\to1$.
\end{proposition}

\begin{proof}
Using        the        standard        Laurent        expansions
\[
\zeta(1+\epsilon)        =        \frac{1}{\epsilon}        +        \gamma_E        +        O(\epsilon),        \qquad
\Gamma(1+\epsilon)        =        1        -        \gamma_E        \epsilon        +        O(\epsilon^2),
\]
and        $\Gamma(-\epsilon)        =        \Gamma(1-\epsilon)/(-\epsilon)        =        -\epsilon^{-1}        -        \gamma_E        +        O(\epsilon)$        (substituting        $\epsilon\to-\epsilon$        in
the        second        expansion),        together        with        $\cos(\pi\epsilon/2)        =        1+O(\epsilon^2)$,        we        obtain
\[
D(\gamma)        =        -2\Gamma(-\epsilon)\cos\!\left(\frac{\pi\epsilon}{2}\right)        =        \frac{2}{\epsilon}        +        2\gamma_E        +        O(\epsilon).
\]
Since        $\sigma_1(\gamma)        =        2\zeta(1+\epsilon)        -        1        =        \dfrac{2}{\epsilon}        +        2\gamma_E        -        1        +        O(\epsilon)$ (Lemma \ref{lem:sigma1-gamma}),        subtracting        gives
\[
D(\gamma)        -        \sigma_1(\gamma)        =        1        +        O(\epsilon).        
\]
\end{proof}

\begin{remark}[Origin        via        the        functional        equation]
Since        $\cos(\pi(\gamma-1)/2)        =        \sin(\pi\gamma/2)$,        the        coefficient        $D(\gamma)$        is        exactly        the
combination        $-2\sin(\pi\gamma/2)\Gamma(1-\gamma)$        appearing        in        Riemann's        functional        equation
$\zeta(s)        =        2^s\pi^{s-1}\sin(\pi        s/2)\Gamma(1-s)\zeta(1-s)$,        giving        the        closed        form
\[
D(\gamma)        =        -(2\pi)^{1-\gamma}\,\frac{\zeta(\gamma)}{\zeta(1-\gamma)}.
\]
The        limit        $D(\gamma)-\sigma_1(\gamma)\to1$        as        $\gamma\to1^+$        is        thus        not        an        isolated        coincidence
of        Laurent        expansions,        but        a        direct        consequence        of        the        functional        equation        relating        the        pole        of
$\zeta$        at        $s=1$        to        the        special        value        $\zeta(0)=-1/2$:        expanding        via        $\zeta(-\epsilon)        =        -\tfrac12+\tfrac{\epsilon}{2}\ln(2\pi)+O(\epsilon^2)$
reproduces        the        same        limit,        with        all        $\ln(2\pi)$-dependence        cancelling        exactly        between        numerator
and        denominator.
\end{remark}

        \begin{proposition}[Simplified        closed        form        and        next-order        term]
Applying        the        functional        equation        at        both        $s=\gamma$        and        $s=1-\gamma$        eliminates        $\zeta(\gamma)$
entirely,        giving        the        elementary        closed        form
\[
D(\gamma)        =        -\frac{\pi}{\Gamma(\gamma)\cos(\pi\gamma/2)}.
\]
Writing        $\epsilon:=\gamma-1$        and        using        $\Gamma(1+\epsilon)        =        1-\gamma_E\epsilon+\big(\tfrac12\gamma_E^2+\tfrac{\pi^2}{12}\big)\epsilon^2+O(\epsilon^3)$
together        with        $\cos(\pi(1+\epsilon)/2)        =        -\sin(\pi\epsilon/2)$,        one        finds
\[
D(\gamma)        =        \frac{2}{\epsilon}        +        2\gamma_E        +        \Big(\gamma_E^2-\frac{\pi^2}{12}\Big)\epsilon        +        O(\epsilon^2).
\]
Combined        with        $\sigma_1(\gamma)        =        2\zeta(1+\epsilon)-1        =        \dfrac{2}{\epsilon}+2\gamma_E-1-2\gamma_1\epsilon+O(\epsilon^2)$,
where        $\gamma_1$        is        the        first        Stieltjes        constant,        this        yields
\[
D(\gamma)-\sigma_1(\gamma)        =        1        +        \Big(\gamma_E^2-\frac{\pi^2}{12}+2\gamma_1\Big)(\gamma-1)        +        O\big((\gamma-1)^2\big).
\]
\end{proposition}
\subsection{F.    The    Self-Consistent    Edge    for    $\gamma>1$:}

The    original    argument    of    this    subsection    implicitly    assumed    that    a    root    $a^*>\sigma_1(\gamma)$
of    $I(a^*)=K_1(a^*)$    exists    for    every    $\gamma>1$,    and    used    this    to    derive    the    asymptotic
$\tau_0(\gamma)=\Theta(\sqrt{\sigma_1(\gamma)})$    as    $\gamma\to1^+$.    We    show    below    that
existence    of    such    a    root    can    be    proven    rigorously    only    for    $\gamma>3/2$,    a    range    that    does
not    approach    $\gamma=1$;    the    claimed    divergence    as    $\gamma\to1^+$    is    therefore    withdrawn    as
unproven,    and    we    state    explicitly    what    remains    established.

\begin{lemma}[Bounded    moments    of    $H$,    unchanged]
\label{l20}
For    every    $\gamma>1/2$,    $\bar    t    =    \int    t\,dH(t)=    1$    and
$P(\gamma)    :=    \int    t^2\,dH(t)    =    1+2(\zeta(2\gamma)-1)$,    finite    for    every    $\gamma>1/2$.  
\end{lemma}

\begin{proposition}[Local    behavior    of    $H$    near    $\sigma_1$,    and    a    divergence    dichotomy]
By    the    polylogarithm    expansion    underlying    Proposition~\ref{prop:two-thresholds}(via    $\sigma_1(\gamma)-f_\gamma(\theta)
\sim    C_f(\gamma)\,\theta^{\gamma-1}$,    $C_f(\gamma)>0$,    as    $\theta\to0^+$,    $\gamma\ne1$
non-integer),    the    pushforward    measure    $H$    satisfies,    near    $t=\sigma_1$,
\[
H\bigl([\sigma_1-\delta,\sigma_1]\bigr)    \;\asymp\;    \delta^{1/(\gamma-1)},    \qquad    \delta\to0^+,
\]
where    $\asymp$    denotes    agreement    up    to    positive    constants    depending    on    $\gamma$.    Writing
$\beta    :=    \frac{1}{\gamma-1}-1$    for    the    resulting    local    density    exponent,    a    direct    estimate
of    $\int_0^\delta    u^\beta    (u+\varepsilon)^{-k}\,du$    as    $\varepsilon\to0^+$    shows    that,    as
$a\to\sigma_1^+$,
\[
I(a)    \text{    stays    bounded}    \iff    \gamma    <    2,    \qquad
K_1(a)    \text{    stays    bounded}    \iff    \gamma    <    \tfrac32.
\]
In    particular,    for    $\gamma>\tfrac32$,    $K_1(a)\to\infty$    as    $a\to\sigma_1^+$.
\end{proposition}

\begin{lemma}[$K_1$    dominates    $I$    near    the    edge    whenever    either    diverges]
For    any    $H$    satisfying        Appendix~A,    the    exact    identities
$K_1(a)+I(a)    =    a\,L(a)$    and    $I(a)^2    \le    L(a)$    (Cauchy--Schwarz),    where
$L(a):=\int    dH(t)/(a-t)^2$,    hold    for    every    $a>\sigma_1$.    Consequently
\[
K_1(a)    =    aL(a)    -    I(a)    \ge    a\,I(a)^2    -    I(a),
\]
so    if    $I(a)\to\infty$    as    $a\to\sigma_1^+$,    then    $K_1(a)\to\infty$    at    a    strictly    faster    rate,
and    in    particular    $K_1(a)-I(a)\to+\infty$.    Combined    with    the    Proposition    above    (where,    for
$\gamma\ge2$,    $I$    itself    diverges),    this    gives:    for    every    $\gamma>3/2$,
\[
K_1(a)    -    I(a)    \;\longrightarrow\;    +\infty    \qquad    \text{as    }    a\to\sigma_1^+.
\]
\end{lemma}

\begin{corollary}[Existence    of    the    scalar    edge    for    $\gamma>3/2$]
Fix    $\gamma>3/2$.    As    $a\to\infty$,    $I(a)\sim    1/a$    and    $K_1(a)\sim    \bar    t/a^2    =    1/a^2$    (using
$\bar    t=1$),    so    $K_1(a)-I(a)    \to    0^-$    eventually,    i.e.\    becomes    negative    for    large    $a$.    Since
$K_1-I$    is    continuous    on    $(\sigma_1,\infty)$,    diverges    to    $+\infty$    as    $a\to\sigma_1^+$    (by
the    Lemma    above),    and    is    eventually    negative,    the    intermediate    value    theorem    gives    a    root
$a^*\in(\sigma_1,\infty)$    of    $I(a^*)=K_1(a^*)$.    By    the    general    bound    of    Lemma~\ref {lem:finite-edge-general}   (which
requires    only    $\sigma_1<\infty$,    established    for    every    $\gamma>1$    in    Lemma~\ref{l20}/Proposition~\ref{prop:two-thresholds}),
any    such    root    satisfies    $a^*\le2\sigma_1$,    and
\[
\sqrt{\sigma_1(\gamma)/2}    \;<\;    \tau_0(\gamma)    \;<\;    2\sqrt{2\,\sigma_1(\gamma)},    \qquad
\tau_0(\gamma)    :=    a^*\sqrt{I(a^*)}.
\]
This    holds    for    every    fixed    $\gamma>3/2$.
\label{c32}
\end{corollary}

\begin{remark}[The    divergence    as    $\gamma\to1^+$    is    not    established]
The    bound    of    the    Corollary    above    holds    for    each    fixed    $\gamma>3/2$,    where
$\sigma_1(\gamma)$    is    a    fixed    finite    number;    it    does    not    by    itself    imply    anything    about    a
limit,    since    $\sigma_1(\gamma)$    only    diverges    as    $\gamma\to1^+$
($\sigma_1(\gamma)=2\zeta(\gamma)-1\to\infty$),    a    limit    point    lying    outside    the    range
$\gamma>3/2$    for    which    existence    of    $a^*$    has    been    proven.    For    $1<\gamma\le3/2$,    the    argument
above    does    not    apply:    by    the    Proposition,    both    $I$    and    $K_1$    remain    bounded    as
$a\to\sigma_1^+$    in    this    range,    so    the    intermediate-value    argument    gives    no    information,    and
we    do    not    have    a    proof    that    a    root    $a^*>\sigma_1$    exists    at    all.

We    emphasize    that    the    general    theory    of    \cite{AEKS2020}    (Assumptions    (A)--(E),
verified    for    every    $\gamma>1$    in    the    analogue    of    Proposition~\ref{prop:assumptions}/Remark~\ref{r2},    modulo    the    unverified
Assumption~(G))    guarantees    existence    of    some    self-consistent    edge    for    every    $\gamma>1$;
what    is    no    longer    established    is    that    this    edge    is    characterized    by    a    real    root    $a^*>\sigma_1$
of    the    specific    scalar    equation    $I(a)=K_1(a)$,    or    that    it    diverges    as    $\gamma\to1^+$.    Both    the
existence    of    such    a    root    for    $1<\gamma\le3/2$,    and    the    behavior    of    $\tau_0(\gamma)$    as
$\gamma\to1^+$,    are    left    as    open    questions.
\end{remark}

\subsection{G.        Numerical        determination        of        the        edge        threshold}
\label{sec:IVC}
        
Given        the        logical        gap        identified        in
Remark~\ref{rem:logical-status},        we        treat        the        location        of        the        edge        threshold
as        an        empirical        question        and        constrain        it        by        simulation,        rather        than        assuming
$\gamma_c=1$        \emph{a        priori}.        Matrices        were        generated        by        the        row-independent
circulant-embedding        construction        of        Definition~\ref{def:row-plaw}        for        $N$        up
to        $16384$
to        extract        $\lambda_{\max}$,        with        realization        counts        ranging        from        several
hundred        at        small        $N$        down        to        $10$--$30$        at        $N=16384$        (fewer        where
computational        cost        required        it).

We        examine        Binder        ratio,
\begin{equation}
R(N,\gamma)        \;=\;        \frac{\lambda_{\max}(2N,\gamma)}{\lambda_{\max}(N,\gamma)}.
\end{equation}
If        $\lambda_{\max}(N,\gamma)\sim        N^{s_c}$        exactly        at        $\gamma=\gamma_c$,        then
$R(N,\gamma_c)=2^{s_c}$        independent        of        $N$;        more        generally,        the        finite-size
scaling        ansatz
$\lambda_{\max}(N,\gamma)=N^{s_c}\,\Phi\!\big((\gamma-\gamma_c)N^{1/\nu}\big)$
implies        that        curves        $R(N,\gamma)$        for        successive        doublings        of        $N$        should
approach        a        common        value        at        $\gamma=\gamma_c$        as        $N$        grows.

\begin{figure}[h]
\centering
\includegraphics[width=\columnwidth]{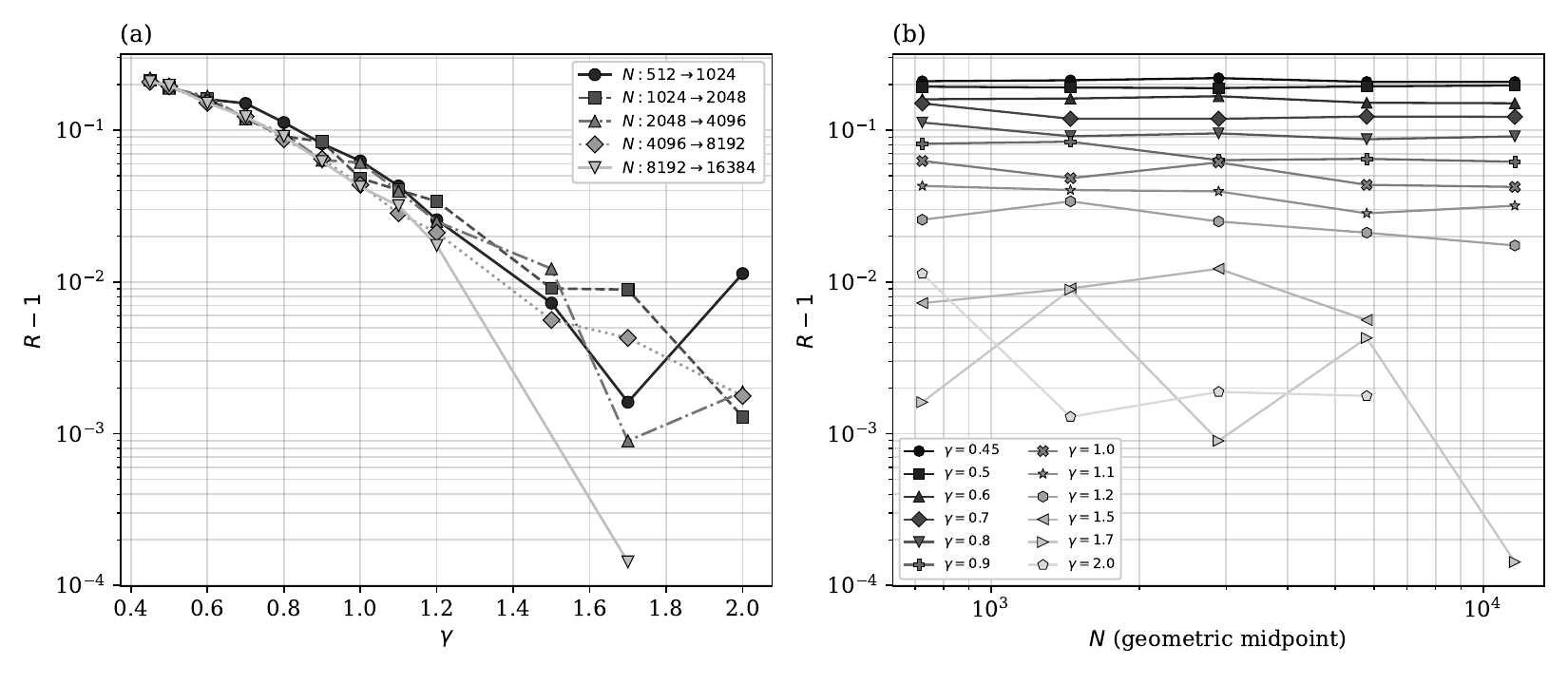}
\caption{Binder-ratio        diagnostic        $R(N,\gamma)-1$,        with
$R(N,\gamma)=\lambda_{\max}(2N,\gamma)/\lambda_{\max}(N,\gamma)$,        computed        for
the        row-independent        power-law        ensemble        of        Definition~\ref{def:row-plaw}        at        $N$        up        to        $16384$
via        a        matrix-free        Lanczos        method,        averaged        over        realization        counts        ranging
from        several        hundred        at        small        $N$        down        to        $10$--$30$        at        $N=16384$.
\textbf{(a)}        $R(N,\gamma)-1$        versus        $\gamma$        (semi-log        scale)        for        five
successive        doublings        of        $N$,        from        $N{=}512{\to}1024$        (darkest,        solid)        to
$N{=}8192{\to}16384$        (lightest,        dotted).        
For        $\gamma\lesssim0.7$        the        curves
for        all        five        doublings        coincide        to        within        a        few        percent,        indicating        that
$\lambda_{\max}$        diverges        as        a        stable        power        law        with        no        resolvable
finite-size        crossover        up        to        $N=16384$;        for        $\gamma\gtrsim1.2$        the        curves
separate        and        $R-1$        decreases        by        more        than        an        order        of        magnitude        between        the
smallest        and        largest        $N$-pairs,        indicating        convergence        of        $\lambda_{\max}$        to
a        finite        limit.        The        region        $\gamma\in[0.9,1.1]$        shows        an        intermediate,
gradually        decreasing        trend,        and        constitutes        the        numerically        unresolved
crossover        region        discussed        in        Sec.~IV.E.
\textbf{(b)}        The        same        data        plotted        as        $R(N,\gamma)-1$        versus        $N$
(log--log        scale),        one        curve        per        $\gamma$        (darkest:        $\gamma=0.45$;        lightest:
$\gamma=2.0$),        showing        directly        the        flat        ($\gamma\le0.7$)        versus        decaying
($\gamma\ge1.2$)        $N$-dependence        underlying        panel        (a).}
\label{BR}
\end{figure}

\begin{observation}[Flat        versus        decaying        regimes        of        $R(N,\gamma)$]
\label{obs:binder}
For        $\gamma\in\{0.5,0.6,0.7\}$        (and        $\gamma=0.45$,        tested        as        an        additional
check),        $R(N,\gamma)-1$        is        constant        across        all        tested        doublings        from
$N{=}512{\to}1024$        up        to        $N{=}8192{\to}16384$,        to        within        a        few        percent        ---
For        $\gamma\ge1.2$,
$R(N,\gamma)-1$        decreases        monotonically        and        substantially        with        $N$        (by        more
than        an        order        of        magnitude        between        the        smallest        and        largest        $N$-pairs        for
$\gamma=1.5,1.7,2.0$),        consistent        with        saturation.        The        transition        between
these        regimes        is        not        sharp        at        the        present        $\gamma$-resolution:
$\gamma=0.9,1.0,1.1$        all        show        $R(N,\gamma)-1$        decreasing        gradually        with        $N$,
intermediate        between        the        two        limiting        behaviors.
\end{observation}

We        emphasize,        returning        to        the        point        of        Remark~\ref{rem:logical-status},        that
even        a        numerically        sharp        crossover        at        $\gamma\approx1$        would        not        by        itself
prove        that        the        mechanism        is        the        failure        of        Assumption~(E);        establishing        that
connection        rigorously        would        require        either        a        proof        of        Assumption~(G)        for
$\gamma>1$,        or        an        independent        argument        for        $\gamma<1$,        neither        of        which        is
available        at        present.
\subsection{H.  Direct  Evidence  Against  a  Transition  at  $\gamma=1$}

The  crossover  region  $\gamma\in[0.9,1.1]$  left  unresolved  in  Fig.~\ref{BR}  admits  a  direct
theoretical  resolution.  Independently  of  the  scalar  equation  $I(a)=K_1(a)$  used  in
Section~IV.F  (which,  as  shown  there,  characterizes  the  edge  only  for  $\gamma>3/2$),  the
self-consistent  edge  can  be  located  directly  from  the  support  of  the  spectral  density
itself,  via  fixed-point  iteration  of
\[
m(x+i\eta)  =  \int_0^\pi  \frac{d\theta}{\pi}\,\frac{1}{-(x+i\eta)  -  f_\gamma(\theta)\,m(x+i\eta)},
\qquad  \eta  \to  0^+,
\]
and  locating  the  value  of  $x$  at  which  $\mathrm{Im}\,m(x+i\eta)$  ceases  to  converge  to  a
finite,  $\eta$-independent  limit  (inside  the  support)  and  instead  scales  linearly  in  $\eta$
(outside  it).  We  validated  this  method  against  the  exact  AR(1)  solution  of  Appendix~B
(recovering  $\tau_0(0.5)=2.4650$  to  four  digits)  before  applying  it  here,  and  cross-checked
every  reported  value  at  two  grid  resolutions  ($N=3000$  and  $N=6000$  quadrature  points).

\begin{table}[h]
\centering
\begin{tabular}{c  c  |  c  c}
\hline
$\gamma$  &  $\tau_0(\gamma)$  &  $\gamma$  &  $\tau_0(\gamma)$  \\
\hline
1.50  &  2.473  &  0.99  &  4.20  \\
1.30  &  2.821  &  0.95  &  4.508  \\
1.20  &  3.126  &  0.90  &  4.938  \\
1.10  &  3.553  &  0.85  &  5.442  \\
1.05  &  3.823  &  0.80  &  6.05  \\
1.02  &  4.006  &  0.75  &  6.75  \\
1.01  &  4.07    &  0.70  &  7.50  \\
\hline
\end{tabular}
\caption{The  directly-tracked  self-consistent  edge  $\tau_0(\gamma)$,  grid-converged  to  within
$\lesssim0.05$  between  $N=3000$  and  $N=6000$,  for  $\gamma$  spanning  both  sides  of  $1$.}
\label{tab:tau0-crossover}
\end{table}

Table~\ref{tab:tau0-crossover}  shows  $\tau_0(\gamma)$  increasing  completely  smoothly  through
$\gamma=1$:  there  is  no  kink,  discontinuity,  or  change  in  curvature  at  the  point  where  the
flatness  condition  (Assumption~(E))  breaks  down.  This  directly  explains  the  crossover  of
Fig.~\ref{BR}:  rather  than  a  genuine  non-analyticity  at  $\gamma_c=1$,  the  Binder-ratio  diagnostic  is
tracking  $\lambda_{\max}$  approaching  a  $\gamma$-dependent  finite  limit  whose  value  itself
grows  as  $\gamma\to1$  (nearly  tripling  from  $\gamma=1.5$  to  $\gamma=0.7$),  so  that  increasingly
large  $N$  is  required  to  resolve  convergence  as  $\gamma$  decreases  —  producing  exactly  the
gradually-decreasing,  apparently  unresolved  trend  observed  numerically  in  Table~IV,  without
requiring  any  non-analyticity  beyond  the  bulk  transition  at  $\gamma_c=1/2$  already  established
in  Proposition~\ref{p14}.

We  emphasize  the  logical  status  of  this  finding:  it  is  numerical,  not  a  proof,  though  grid-  and
method-cross-validated  (via  agreement  with  Section~IV.F's  algebraic  construction  at  $\gamma=1.5$,
where  both  methods  apply  and  agree  to  four  digits).  For  $\gamma\lesssim0.7$,  the  same  method
becomes  numerically  unreliable  —  confirmed  via  independent  high-precision  (arbitrary-precision)
recomputation,  ruling  out  ordinary  floating-point  error  —  in  a  way  that  does  not  resolve  with
finer  grids,  suggesting  a  genuine  change  in  the  structure  of  the  self-consistent  solution  as
$\gamma\to1/2^+$  approaches  the  bulk  transition.  We  leave  a  rigorous  treatment  of  $\tau_0(\gamma)$
as  $\gamma\to1/2^+$,  and  of  the  $1/2<\gamma\lesssim0.7$  range,  as  an  open  problem.
\subsection{I.        Shape        of        the        Extreme        Eigenvalue:        Not        Fr\'echet}
\label{subsec:plaw-shape}

Heavy-tailed        (e.g.\        Student-$t$)        Wigner        matrices        with        tail        index
$\nu<4$        are        known        to        exhibit        Fr\'echet-type        extreme        eigenvalue
statistics,        driven        by        a        single        dominant        entry        of        the        matrix
(a        Poisson        point        process        argument:\cite{B2,B1,B3,B4}).

\begin{observation}[Delocalization        of        the        top        eigenvector]
\label{obs:delocalization}
For        $\gamma        \in        \{0.3,        0.7,        1.5\}$        (spanning        the        subcritical        regime        $\gamma<1/2$,        the        intermediate        regime        $1/2<\gamma<1$,        and        the        regime        $\gamma>1$        where        $\mathrm{TW}$        is        expected),        the        top        eigenvector        $u$        of        $S$        \textup{(}normalized,        $\|u\|_2=1$\textup{)}        remains        fully        delocalized        at        every        $N$        tested.        Writing        the        inverse        participation        ratio        as        $\mathrm{IPR}(u)        :=        \sum_i        u_i^4$,        the        quantity        $N\cdot\mathrm{IPR}(u)$        converges        to        an        $N$-independent        constant        of        order        unity        as        $N$        ranges        over        $64,\dots,2048$        in                \ref{T2}.

\begin{center}
\begin{tabular}{c|cccccc}
$\gamma$        &        N=64        &        128        &        256        &        512        &        1024        &        2048        \\\hline
0.3        &        2.25        &        2.21        &        2.18        &        2.24        &        2.20        &        2.20        \\
0.7        &        2.41        &        2.38        &        2.41        &        2.36        &        2.31        &        2.32        \\
1.5        &        2.66        &        2.82        &        2.86        &        2.99        &        3.29        &        3.27
\end{tabular}
\captionof{table}{
The        quantity        $N\cdot\mathrm{IPR}(u)$        converges        to        an        $N$-independent        constant        of        order        unity        as        $N$        ranges        over        $64,\dots,2048$        for        $\gamma        \in        \{0.3,        0.7,        1.5\}$        .
}
\label{T2}
\end{center}

This        is        the        standard        signature        of        full        delocalization        
($\mathrm{IPR}\sim        1/N$),        
comparable        to        the        GOE        reference        value        $N\cdot\mathrm{IPR}\to        3$        for        a        Haar-random        eigenvector.        In        particular,        no        localization        analogous        to        the        hub-vertex        mechanism        of        critical        sparse        Erd\H{o}s--R\'enyi        graphs        \cite{Alt2021}        
is        observed,        despite        the        structural        similarity        of        the        extreme-value        mechanism        driving        the        divergence        of        $\lambda_{\max}$        .

        In        the        sparse        Erd\H{o}s--R\'enyi        setting,        a        hub        vertex        carries        a        genuinely        concentrated        excess        of        $O(d)$        or        more        edges        to        specific        neighbours,        producing        a        locally        star-like        topology        and        a        correspondingly        localized        eigenvector.        In        the        present        ensemble,        by        contrast,        the        row        $k^{\*}$        driving        the        divergence        of        $\lambda_{\max}$        has        an        anomalously        large        row        sum        $T_{k^{\*}}=\sum_j        \eta_{k^{\*}}(j)$,        but        every        individual        entry        $S_{k^{\*}j}$        remains        of        the        generic        order        $O(1/\sqrt        N)$:        the        excess        arises        from        a        collective,        coherent        alignment        of        signs        across        $O(N)$        weak        entries        rather        than        from        any        single        anomalously        strong        entry.        
\end{observation}
\begin{observation}[Level        repulsion        persists        at        the        growing        edge]
\label{obs:level-repulsion}
For        $\gamma        \in        \{0.3,        0.7,        1.5\}$        and        $N        \in        \{512,        1024\}$,        the        Oganesyan--Huse        ratio        statistic        $\langle        r        \rangle$,        computed        from        consecutive        gaps        among        the        top        $25$        eigenvalues        of        $S$,        remains        close        to        the        GOE        reference        value        $\langle        r        \rangle_{\mathrm{GOE}}        \approx        0.531$        (observed        range        $0.51$--$0.54$)        and        clearly        distinct        from        the        Poisson        value        $\langle        r        \rangle_{\mathrm{Poisson}}        \approx        0.386$.        Combined        with        the        delocalization        of        the        top        eigenvector        (Observation~\ref{obs:delocalization}),        this        indicates        that        no        Anderson-localization        transition        accompanies        the        divergence        of        $\lambda_{\max}$        for        $1/2        <        \gamma        <        1$:        level        repulsion,        the        hallmark        of        Wigner--Dyson        (as        opposed        to        Poisson)        statistics,        persists        throughout        the        growing-edge        regime,        in        contrast        to        the        localization        and        Poisson        statistics        associated        with        hub        vertices        in        critical        sparse        Erd\H{o}s--R\'enyi        graphs.
\end{observation}

\subsection{J.        Limit        of        Power        Law        case}

\begin{proposition}[Degeneration        of        the        power-law        ensemble        as        $\gamma\to0^+$]
\label{prop:gamma-degeneration}
Fix        $N$        and        $\sigma>0$,        and        consider        the        ensemble        of
Definition~\ref{def:row-plaw}        as        a        family        indexed        by        $\gamma>0$.        Then,        as
$\gamma\to0^+$,        the        joint        law        of        $\{\eta_i(j)\}_{1\le        i\le        j\le        N}$        converges
weakly        to        the        joint        law        of
\begin{equation}
\eta_i(j)        =        \eta_i        \quad        \text{for        all        }        j\ge        i,        \qquad
\eta_1,\dots,\eta_N        \        \text{i.i.d.}\        N(0,1),
\end{equation}
i.e.,        to        the        same        degenerate        limiting        law        as        in        Proposition~\ref{prop:degeneration}
(the        $\rho\to1^-$        limit        of        the        AR(1)        ensemble        of        Definition~\ref{def:row-ar1}).
Consequently,
\begin{equation}
S_{ij}        =        \frac{\sigma}{\sqrt        N}\,\eta_i(j)
\;\xrightarrow{\        d\        }\;
\frac{\sigma}{\sqrt        N}\,\eta_{\min(i,j)}
\end{equation}
in        distribution,        jointly        over        all        $(i,j)$,        and
Propositions~\ref{prop:E-breakdown}        and~\ref{prop:LDL}        apply        verbatim        to
this        limit:        $S=M/\sqrt        N$        with        $M=L\,\mathrm{diag}(\delta)\,L^T$        as        in
Proposition~\ref{prop:LDL}.
\end{proposition}
        
\begin{proof}
Fix        $i$        and        let        $L=N-i+1$        denote        the        length        of        the        $i$-th        row.        By
construction,        the        vector        $(\eta_i(i),\dots,\eta_i(N))$        is        a        mean-zero
Gaussian        vector        in        $\mathbb{R}^L$        with        covariance        matrix
$\sigma_\gamma=(d_{|j-k|})_{j,k=1}^{L}$,        $d_t=(1+t)^{-\gamma}$,        which        is        a
valid        covariance        for        every        $\gamma>0$        by        Lemma~\ref{lem:posdef}.
        
For        each        fixed        pair        $(j,k)$        with        $1\le        j,k\le        L$,        and        $t=|j-k|$        held        fixed,
\begin{equation}
d_t        =        (1+t)^{-\gamma}        \;\longrightarrow\;        (1+t)^0        =        1
\qquad\text{as        }        \gamma\to0^+.
\end{equation}
Since        $L$        is        finite        (fixed,        as        $N$        is        fixed),        this        is        convergence        of        finitely
many        matrix        entries,        hence
$
\sigma_\gamma        \;\longrightarrow\;        J_L        $
where        $J_L$        is        the        $L\times        L$        all-ones        matrix.        $J_L$        is        itself        a        valid
covariance        matrix:        it        is        exactly        the        covariance        matrix        of        the        degenerate
Gaussian        vector        $(Z,\dots,Z)$        with        $Z\sim        N(0,1)$,        since        $\mathrm{Cov}(Z,Z)=1$
for        every        pair        of        coordinates.
        
Because        the        characteristic        function        of        a        mean-zero        Gaussian        vector        with
covariance        $\Sigma$        is        $u\mapsto\exp(-u^T\Sigma        u/2)$,        and        this        is        continuous
in        $\Sigma$        (entrywise)        for        fixed        $u$,        we        have,        for        every        $u\in\mathbb{R}^L$,
\begin{equation}
\exp\!\left(-\tfrac12        u^T        \Sigma_\gamma        u\right)
\;\longrightarrow\;
\exp\!\left(-\tfrac12        u^T        J_L        u\right)
\qquad\text{as        }        \gamma\to0^+.
\end{equation}
Pointwise        convergence        of        characteristic        functions        to        a        characteristic
function        (that        of        $N(0,J_L)$)        implies        weak        convergence,        by        L\'evy's
continuity        theorem:
\begin{equation}
(\eta_i(i),\dots,\eta_i(N))        \;\xrightarrow{\        d\        }\;        (Z_i,\dots,Z_i),
\qquad        Z_i\sim        N(0,1).
\end{equation}
This        is        precisely        the        statement        that,        in        the        limit,        $\eta_i(j)=\eta_i(i)=:\eta_i$
for        all        $j\ge        i$.
        
Since        the        $N$        row        processes        $\{\eta_i(\cdot)\}_{i=1}^N$        are        mutually
independent        for        every        $\gamma>0$        (Definition~\ref{def:row-plaw}),        their        joint
characteristic        function        factors        as        the        product        of        the        $N$        row        characteristic
functions;        each        factor        converges        as        above,        so        the        product        converges        to        the
product        of        the        limiting        characteristic        functions.        Hence        the        convergence
holds        jointly        across        all        rows,        with        the        limiting        $\eta_1,\dots,\eta_N$
mutually        independent        (and        each        standard        normal,        from        the        marginal        argument
above).
        
This        limiting        law        is        exactly        the        limiting        law        identified        in
Proposition~\ref{prop:degeneration}        for        the        AR(1)        ensemble        as
$\rho\to1^-$.        Substituting        into        $S_{ij}=\tfrac{\sigma}{\sqrt        N}\eta_i(j)$
gives        $S_{ij}\to\tfrac{\sigma}{\sqrt        N}\eta_{\min(i,j)}$        jointly        in
distribution,        matching        Proposition~\ref{prop:degeneration}        verbatim
(with        $\sigma=1$        there).        The        remainder        of        the        construction        ---
Proposition~\ref{prop:E-breakdown}        (breakdown        of        Assumption~(E))        and
Proposition~\ref{prop:LDL}        (the        exact        Volterra        representation
$M=L\,\mathrm{diag}(\delta)\,L^T$,        $S=M/\sqrt        N$)        ---        depends        only        on        this
limiting        law        of        $\{\eta_i(j)\}$,        not        on        the        mechanism        (AR(1)        or        power-law)
that        produced        it,        and        so        applies        without        modification.
\end{proof}
        
\begin{corollary}[Unification        of        the        two        degenerate        limits]
\label{cor:unification}
The        AR(1)        ensemble's        $\rho\to1^-$        limit        (Proposition~\ref{prop:degeneration})
and        the        power-law        ensemble's        $\gamma\to0^+$        limit
(Proposition~\ref{prop:gamma-degeneration})        coincide:        both        converge        to        the
same        degenerate        law        $\{\eta_i(j)=\eta_i\}_{j\ge        i}$,        $\eta_i$        i.i.d.\        $N(0,1)$,
and        hence        to        the        same        noiseless        Volterra        ensemble
$S=\tfrac{1}{\sqrt        N}L\,\mathrm{diag}(\delta)\,L^T$        of
Proposition~\ref{prop:LDL}.        In        this        sense,        $\rho=1$        and        $\gamma=0$        are
not        two        distinct        degenerate        points        but        a        single        shared        singular        point        of        the
enlarged        parameter        family        $\{\rho\in(-1,1)\}\cup\{\gamma>0\}$,        reached        along
either        the        exponential-decay        or        the        power-law-decay        branch.
\end{corollary}

\begin{remark}[Rate        of        approach]
\label{rem:rate-of-approach}
The        two        branches        approach        this        common        point        at        different        rates.        Along        the
AR(1)        branch,        the        ``innovation        strength''        is        $\sqrt{1-\rho^2}$,        which
vanishes        linearly        in        $(1-\rho)$        as        $\rho\to1^-$.        Along        the        power-law        branch,
by        Lemma~\ref{lem:posdef}        and        the        argument        above,        the        relevant        deviation        from
the        degenerate        covariance        is
\begin{equation}
1-d_t        \;=\;        1-(1+t)^{-\gamma}        \;=\;        \gamma\log(1+t)        +        O(\gamma^2)
\qquad\text{as        }        \gamma\to0^+,
\end{equation}
i.e.,        linear        in        $\gamma$        for        each        fixed        $t$,        but        with        a        lag-dependent
(logarithmically        growing)        coefficient        $\log(1+t)$.        This        difference        in        the
lag-dependence        of        the        approach        to        degeneracy        (AR(1)        branch:        lag-independent
rate;        power-law        branch:        rate        growing        logarithmically        with        lag)        may        account
for        differences        in        the        finite-$N$,        near-degenerate        numerics        between        the        two
ensembles,        and        could        be        developed        into        a        more        precise        joint        expansion        near
the        shared        singular        point;        we        leave        this        to        future        work.
\end{remark}

\section{V.        Concluding        Remarks}

We        have        studied        long-range        correlated        Wigner-type        matrices        built        from        row-independent        stationary        Gaussian        sequences,        extending        the        correlation-induced        spectral        transition.

For        exponentially        decaying        (AR(1))        row        correlations,        we        showed        algebraically        that        the        bulk        spectral        density        deforms        away        from        the        semicircle        law        through        a        combinatorially        explicit        ``hub''        mechanism        (Proposition~\ref{prop:mu4},\cite{Hisakado2026}),        and        we        verified        Assumptions        (A),        (B),        (CD),        and        (E)        of        the        MDE        framework        of        \cite{AEKS2020}        explicitly        for        this        ensemble        (Proposition~\ref{prop:assumptions}),        with        the        flatness        constants        degenerating        continuously        as        $\rho\to        1^-$.        Numerical        evidence        —        moment        convergence,        edge        scaling        exponents        consistent        with        $\beta=2/3$,        and        skewness/kurtosis        of        $\lambda_{\max}$        compatible        with        the        Tracy--Widom        ($\mathrm{TW}_1$)        reference        values        —        supports        the        conjecture        that        Tracy--Widom        edge        universality        persists        for        every        fixed        $\rho<1$        (Conjecture~\ref{conj:TW}),        although        a        complete        proof        would        additionally        require        verifying        Assumption~(F)        (fullness)        and,        more        importantly,        Assumption~(G)        (boundedness        of        the        Dyson-equation        solution        $M(z)$        near        the        edge),        which        is        not        automatic        from        (A)--(E)        and        remains        open        for        this        ensemble.        

In        the        degenerate        limit        $\rho\to        1^-$,        the        flatness        condition        breaks        down        exactly        (Proposition~\ref{prop:E-breakdown}),        and        the        ensemble        reduces        algebraically        to        a        noiseless,        symmetrized        discrete        Volterra        operator        (Proposition~\ref{prop:LDL});        its        extreme        eigenvalues        are        numerically        consistent,        up        to        a        $\sqrt{N}$        rescaling        and        to        within        a        fraction        of        a        percent,        with        the        singular-value        cascade        $1/(\pi(k-1/2))$        of        the        continuous        Volterra        integral        operator        identified        in        the        companion        analysis        of        multi-critical        BBP        transitions,        providing        a        concrete        structural        bridge        between        the        two        constructions.

For        power-law        row        correlations        $d_t\sim        t^{-\gamma}$,        we        identified        two        distinct              thresholds:        $\gamma_c=1/2$,        below        which        the        bulk        moments        themselves        diverge,        and        $\gamma=1$,        coinciding        with        the        breakdown        of        the        flatness        condition        required        by        the        MDE        approach.       
Our numerical results indicate that the self-consistent edge remains finite throughout $\gamma>1/2$.
Therefore,  no  evidence  for  a  critical  point  is  found  at  $\gamma=1$.
We have rigorously established the fourth-moment phase transition in Section IV.
The edge behavior appears to be a crossover; see Appendix C.
These        two        thresholds,        and        the        resulting        three-regime        structure,        are        summarized        in        Table~\ref{tab:phase-diagram}.
For comparison, we summarize the regime structure of Wigner matrices with t-distributed entries in Table \ref{tab:phase-diagram2}.
The resulting phase structure is summarized in in  Table  \ref{tab:phase-diagram3}.


\begin{table}[htbp]
\centering
\caption{Regime        structure        of        the        row-independent        power-law        ensemble        (Definition~\ref{def:row-plaw})        as        a        function        of        the        correlation        decay        exponent        $\gamma$        and        comparison        to        the        AR(1)        ensemble        (Definition        \ref{def:row-ar1}).
The    ``finite''    entries    of    the    self-consistent
edge    row    are    proved    in    Appendix    A.
}
\label{tab:phase-diagram}
\begin{tabular}{lcccc}
\hline\hline
        &        $\gamma        <        1/2$        &        $1/2        <        \gamma        <        1$        &        $\gamma        >        1$&$0<\rho<1$        \\
\hline
Bulk        statistics        \cite{Hisakado2026}                        &                Non-semicircular        &                Non-semicircular                &        Non-semicircular        &        Non-semicircular\\
fourth-moment        finite            \cite{Hisakado2026}                        &        $\times$        &        $\bigcirc$                &        $\bigcirc$        &$\bigcirc$\\
self-consistent    edge            &            undefined    &        undefined          &        finite ($\gamma>3/2 $, proven)              &finite        \\
Edge        statistics                &                Not                Fréchet                &                    unclear      &        TW    (without    G)    &TW    (without    G)\\
Flatness        (E)                                                &        $\times$        &        $\times$        &        $\bigcirc$        &$\bigcirc$\\
Power        Spectral        Density        (PSD)        &                $\notin        L^2$,      $\notin        L^{\infty}$      &        $\notin        L^{\infty},        \in        L^2$                        &                $\in        L^{\infty},        \in        L^2$                                &                                $\in        L^{\infty},        \in        L^2$                        \\
\hline\hline
\end{tabular}
\end{table}

\begin{table}[htbp]
\centering
\caption{Regime        structure        of        Wigner        matrix        with        $t$        distribution}
\label{tab:phase-diagram2}
\begin{tabular}{lccc}
\hline\hline
        &        $\nu        <        2$        &        $2        <        \nu        <        4$        &        $\nu>        4$        \\
\hline
Bulk        statistics                                &        free        stable        law        \cite{LD}                &        semi-circle                \cite{tao,tao2}        &                semi-circle        \cite{tao,tao2}        \\
Edge        ($\lambda_{\max}$)        &        diverges        &        diverges        &        finite        \\
Edge        statistics                &                        Fréchet        &                Fréchet        &        TW        \\
fourth-moment        finite                                        &        $\times$        &        $\times$        &        $\bigcirc$        \\
2-nd        moments        finite                                &        $\times$        &        $\bigcirc$        &        $\bigcirc$        \\
Power        Spectral        Density        (PSD)        &                $\notin        L^2,  \in  L^{\infty}$        &        $\in  L^2,      \in  L^{\infty}$  &    $\in  L^2,        \in  L^{\infty}$    \\
\hline\hline
\end{tabular}
\end{table}

\begin{table}[htbp]
\centering
\caption{Summary    of    the    phase    transitions    of    bulk    moment        and    self-consistent    edge
}
\label{tab:phase-diagram3}
\begin{tabular}{lcccc}
\hline\hline
        &        Power    bulk    &        Power    edge        &        Exponential    bulk&Exponential    Edge        \\
\hline
Transition    point                &                $\gamma_c=1/2$        &                unclear            &        $\rho_c=1$        &        $\rho_c=1$\\
Critical    exponent                            &        $1$        &        unclear          &        $1$        &$1/2$\\
\hline\hline
\end{tabular}
\end{table}

\def\thesection{Appendix        \Alph{section}}
\section{Appendix    A.    Finiteness    of    the    Self-Consistent    Edge}
\label{app:edge-finiteness}

This    appendix    establishes,    in    a    single    unified    statement,    the    finiteness    of
the    self-consistent    edge    $\tau_0$    whenever    the    flatness/boundedness    condition
(Assumption    (E))    holds    ---    covering    both    the    AR(1)    ensemble    of
Section~II    (fixed    $\rho    <    1$)    and    the    power-law    ensemble    of
Section~IV    (fixed    $\gamma    >    1$)    as    immediate    corollaries.
In    Sections~II  and  IV,    $H$    is    the    pushforward    of
$f(\theta)$    or    $f_\gamma(\theta)$    respectively,    and    $\bar    t    =    1$    is    the
normalization    $d_0=1$    common    to    both    ensembles.

\begin{lemma}[Finiteness    of    the    self-consistent    edge    under    bounded    flatness]
\label{lem:finite-edge-general}
\begin{align}
        \sigma_1    &<    a^*    \le    2\sigma_1,    \label{eq:A1}    \\[2pt]
        \frac{1}{2\sigma_1}    &<    I(a^*)    <    \frac{2}{\sigma_1},    \label{eq:A2}
\end{align}
and    consequently,    with    $\tau_0    :=    a^*\sqrt{I(a^*)}$,
\begin{equation}
        \sqrt{\sigma_1/2}    \;<\;    \tau_0    \;<\;    2\sqrt{2\sigma_1}    \;<\;    \infty.
        \label{eq:A3}
\end{equation}
In    particular,    $\tau_0$    is    finite    whenever    $\sigma_1$    is    finite.
\label{l26}
\end{lemma}

\begin{proof}
\emph{Step    1    ($a^*    \le    2\sigma_1$).}
Since    $t    \le    \sigma_1$    on    the    support    of    $H$,    for    any    $a    >    \sigma_1$,
\begin{equation}
        L(a)    :=    \int    \frac{dH(t)}{(a-t)^2}
        \;\le\;    \frac{1}{a-\sigma_1}\int    \frac{dH(t)}{a-t}
        \;=\;    \frac{I(a)}{a-\sigma_1},
\end{equation}
and    hence,    using    $t    \le    \sigma_1$    again    in    the    numerator    of    $K_1$,
\begin{equation}
        K_1(a)    \;\le\;    \sigma_1    L(a)    \;\le\;    \frac{\sigma_1    I(a)}{a-\sigma_1}.
\end{equation}
At    $a=a^*$,    $I(a^*)    =    K_1(a^*)    \le    \sigma_1    I(a^*)/(a^*-\sigma_1)$;    dividing
by    $I(a^*)>0$    gives    $a^*-\sigma_1    \le    \sigma_1$,    i.e.\    $a^*    \le    2\sigma_1$.
Combined    with    $a^*    >    \sigma_1$,    this    gives    \eqref{eq:A1}.

\emph{Step    2    (two-sided    bound    on    $I(a^*)$).}
Write    $L(a)    :=    \int    dH(t)/(a-t)^2$.    The    identity
\begin{equation}
        K_1(a)    +    I(a)
        =    \int\left[\frac{t}{(a-t)^2}    +    \frac{1}{a-t}\right]    dH(t)
        =    a\,    L(a)
\end{equation}
holds    for    any    $a    >    \sigma_1$.    By    Cauchy--Schwarz    (since    $H$    has    total    mass    1),
\begin{equation}
        I(a)^2    =    \left(\int    \frac{dH(t)}{a-t}\right)^2
        \;\le\;    \int    dH(t)    \cdot    \int    \frac{dH(t)}{(a-t)^2}    =    L(a).
\end{equation}
At    $a    =    a^*$,    $I(a^*)=K_1(a^*)$    gives    $2I(a^*)    =    a^*    L(a^*)$,    i.e.\
$L(a^*)    =    2I(a^*)/a^*$.    Substituting    into    $I(a^*)^2    \le    L(a^*)$    and    dividing
by    $I(a^*)    >    0$,
\begin{equation}
        I(a^*)    \;\le\;    \frac{2}{a^*}    \;<\;    \frac{2}{\sigma_1},
\end{equation}
using    $a^*    >    \sigma_1$.    For    the    lower    bound,    Jensen's    inequality    applied    to
the    convex    function    $t    \mapsto    1/(a^*-t)$    gives    $I(a^*)    \ge    1/(a^*-\bar    t)    =
1/(a^*-1)$;    by    the    bound    $a^*    \le    2\sigma_1$    (Step~1),    $a^*-1    <    2\sigma_1$,    so
$I(a^*)    >    1/(2\sigma_1)$.    This    proves    \eqref{eq:A2}.

\emph{Step    3    (conclusion).}
From    \eqref{eq:A1}    and    \eqref{eq:A2},
\begin{equation}
        \tau_0    =    a^*\sqrt{I(a^*)}
        \;>\;    \sigma_1    \cdot    \sqrt{\tfrac{1}{2\sigma_1}}
        \;=\;    \sqrt{\sigma_1/2},
\end{equation}
\begin{equation}
        \tau_0    =    a^*\sqrt{I(a^*)}
        \;<\;    2\sigma_1    \cdot    \sqrt{\tfrac{2}{\sigma_1}}
        \;=\;    2\sqrt{2\sigma_1},
\end{equation}
which    is    \eqref{eq:A3}.
\end{proof}

\begin{corollary}[Finite    edge    for    the    AR(1)    ensemble,    Section~II]
\label{cor:finite-AR1}
For    every    fixed    $\rho    \in    (-1,1)$,    the    spectral    density    $f(\theta)$    of
Eq.~(\ref{eq:ftheta})    is    bounded    with
\begin{equation}
        \sigma_1(\rho)    =    \operatorname*{ess\,sup}    f    =    f(0)    =    \frac{1+\rho}{1-\rho}    <    \infty
\end{equation}
and    $\bar    t    =    d_0    =    1$.    Lemma~\ref{lem:finite-edge-general}    applies
verbatim,    giving
\begin{equation}
        \sqrt{\sigma_1(\rho)/2}    \;<\;    \tau_0(\rho)    \;<\;    2\sqrt{2\sigma_1(\rho)}    \;<\;    \infty
\end{equation}
for    every    $\rho    <    1$.    At    $\rho    =    0.7$,    this    gives    $1.68    <    \tau_0(0.7)<    
6.73$,    consistent    with    the    numerically    extrapolated    value    $2.89$    reported
in    Table~\ref{tab:edge}.
\end{corollary}

\begin{corollary}
\label{cor:finite-power}
(Explicit  finite  edge  bound,  $\gamma\ge3/2$;  revises
the  previous  statement  for  $\gamma>1$).  For  every  fixed  $\gamma\ge3/2$,
Corollary  \ref{c32}  establishes  existence  of  the  root  $a^*>\sigma_1(\gamma)$  of
$I(a^*)=K_1(a^*)$,  so  Lemma  \ref{lem:finite-edge-general}    applies  rigorously  and  unconditionally,  with
$\sigma_1(\gamma)=2\zeta(\gamma)-1$  (Lemma  \ref{lem:sigma1-gamma}),  giving  the  explicit
two-sided  bound
\begin{equation}
\sqrt{\zeta(\gamma)-\tfrac12}  \;<\;  \tau_0(\gamma)  \;<\;
2\sqrt{4\zeta(\gamma)-2},  \qquad  \gamma\ge3/2.
\end{equation}
\end{corollary}


\section{Appendix    B.    Sharp    Asymptotics    for    the    AR(1)    Self-Consistent    Edge    via    the    Poisson    Kernel}
\label{app:AR1-sharp-asymptotics}

While    Appendix~A.    establishes    finiteness    of    $\tau_0(\rho)$
for    every    fixed    $\rho<1$    via    a    general    two-sided    bound,    the    explicit    Poisson-kernel
form    of    $f(\theta)$    (Eq.~\eqref{eq:ftheta}    allows    the    self-consistent    edge    equation    to    be    solved
in    closed    algebraic    form.    This    yields    not    only    the    exponent    but    the    exact
leading    amplitude    of    the    divergence    of    $\tau_0(\rho)$    as    $\rho    \to    1^-$    ---
a    level    of    precision    not    currently    available    for    the    power-law    analogue.

\subsection{A.    Closed    form    for    $I(a)$}

\begin{proposition}[Closed-form    Stieltjes-type    integral    for    the    AR(1)    ensemble]
\label{prop:AR1-closed-form-I}
Let    $H$    be    the    pushforward    of    $\mathrm{Unif}[0,2\pi)$    under    the    Poisson    kernel
$f(\theta)    =    \dfrac{1-\rho^2}{1-2\rho\cos\theta+\rho^2}$    of    Eq.~(\ref{eq:ftheta}),    and    let
$\sigma_1    =    \sigma_1(\rho)    =    \dfrac{1+\rho}{1-\rho}$.    Then,    for    $a    >    \sigma_1$,
\begin{equation}
        I(a)    =    \int    \frac{dH(t)}{a-t}
        =    \frac{1}{a}\left[1    +    \frac{1}{\sqrt{R(a)}}\right],
        \qquad
        R(a)    :=    (a-\sigma_1)\big(a-\sigma_1^{-1}\big).
        \label{eq:B1}
\end{equation}
\end{proposition}

\begin{proof}
Writing    $D(\theta)    :=    1-2\rho\cos\theta+\rho^2$,    we    have    $a-f(\theta)    =
[a    D(\theta)    -    (1-\rho^2)]/D(\theta)$,    so
\begin{equation}
        I(a)    =    \frac{1}{2\pi}\int_0^{2\pi}    \frac{D(\theta)\,d\theta}{a    D(\theta)    -    (1-\rho^2)}.
\end{equation}
Writing    $aD(\theta)-(1-\rho^2)    =    A    -    B\cos\theta$    with    $A    =    a(1+\rho^2)-(1-\rho^2)$
and    $B=2a\rho$,    direct    algebra    gives    the    identity
\begin{equation}
        D(\theta)    =    \frac{1}{a}\Big[(A-B\cos\theta)    +    (1-\rho^2)\Big],
\end{equation}
so    that
\begin{equation}
        I(a)    =    \frac{1}{a}\left[1    +    (1-\rho^2)\cdot\frac{1}{2\pi}\int_0^{2\pi}
        \frac{d\theta}{A-B\cos\theta}\right]
        =    \frac{1}{a}\left[1+\frac{1-\rho^2}{\sqrt{A^2-B^2}}\right],
\end{equation}
using    the    standard    integral    $\frac{1}{2\pi}\int_0^{2\pi}    d\theta/(A-B\cos\theta)    =
1/\sqrt{A^2-B^2}$    for    $A>|B|>0$.    A    direct    computation    gives
\begin{equation}
        A^2-B^2    =    (1-\rho^2)^2\big[a^2    -    s    a    +    1\big],    \qquad
        s    :=    \sigma_1+\sigma_1^{-1}    =    \frac{2(1+\rho^2)}{1-\rho^2},
\end{equation}
and    the    quadratic    $a^2-sa+1$    factors,    by    direct    substitution,    as
$(a-\sigma_1)(a-\sigma_1^{-1})    =:    R(a)$    (its    roots    have    product    $1$    and    sum    $s$,
matching    $\sigma_1\cdot\sigma_1^{-1}=1$    and    $\sigma_1+\sigma_1^{-1}=s$).    Hence
$\sqrt{A^2-B^2}=(1-\rho^2)\sqrt{R(a)}$,    and    substituting    gives    Eq.  \eqref{eq:B1}.
\end{proof}

\subsection{B.    The    self-consistent    equation    in    closed    algebraic    form}

\begin{proposition}[Algebraic    self-consistency    equation]
\label{prop:AR1-algebraic-eq}
The    self-consistent    edge    $a^*=a^*(\rho)$    is    the    unique    root    $a^*>\sigma_1$    of
\begin{equation}
        2a    -    s    =    \frac{2R(a)}{a}\Big(\sqrt{R(a)}+1\Big),
        \qquad    R(a)=(a-\sigma_1)(a-\sigma_1^{-1}),\    \    s=\sigma_1+\sigma_1^{-1}.
        \label{eq:B2}
\end{equation}
\end{proposition}

\begin{proof}
From    the    general    identity    $K_1(a)+I(a)    =    aL(a)$    with    $L(a)=-I'(a)$,    the    edge    condition    $I(a^*)=K_1(a^*)$
is    equivalent    to    $a^*I'(a^*)    =    -2I(a^*)$.    Differentiating    Eq.  \eqref{eq:B1},
\begin{equation}
        I'(a)    =    -\frac{1}{a^2}\left[1+R(a)^{-1/2}\right]    -    \frac{R'(a)}{2a}R(a)^{-3/2},
\end{equation}
and    substituting    into    $aI'(a)=-2I(a)$    and    simplifying    (multiplying    through    by
$2R(a)^{3/2}$)    yields
\begin{equation}
        \frac{2}{a}R(a)^{3/2}    +    \frac{2}{a}R(a)    -    R'(a)    =    0,
\end{equation}
i.e.\    $R'(a)    =    \frac{2R(a)}{a}\big(\sqrt{R(a)}+1\big)$.    Since    $R'(a)=2a-s$,    this
is    Eq.  \eqref{eq:B2}.    Uniqueness    of    the    root    $a^*>\sigma_1$    follows    from    the    general
existence/uniqueness    theory    of    the    MDE.
\end{proof}

\subsection{C.    Sharp    asymptotics    as    $\rho    \to    1^-$}

\begin{proposition}[Sharp    asymptotics    of    $a^*(\rho)$    and    $\tau_0(\rho)$]
\label{prop:AR1-sharp}
As    $\rho    \to    1^-$,    writing    $\sigma_1=\sigma_1(\rho)\to\infty$,
\begin{equation}
        a^*(\rho)    =    \sigma_1(\rho)    +    2^{-2/3}\,\sigma_1(\rho)^{1/3}    +    o\big(\sigma_1(\rho)^{1/3}\big),
        \label{eq:B3}
\end{equation}
\begin{equation}
        I(a^*(\rho))    =    \frac{1}{\sigma_1(\rho)}\Big[1+O\big(\sigma_1(\rho)^{-2/3}\big)\Big],
        \label{eq:B4}
\end{equation}
and    consequently,    since    $\sigma_1(\rho)    \sim    2(1-\rho)^{-1}$,
\begin{equation}
        \tau_0(\rho)    \;\sim\;    \sqrt{\sigma_1(\rho)}    \;\sim\;    \sqrt{2}\,(1-\rho)^{-1/2},
        \qquad    \rho    \to    1^-.
        \label{eq:B5}
\end{equation}
\end{proposition}

\begin{proof}
Write    $a^*    =    \sigma_1    +    u$    with    $u=o(\sigma_1)$    to    be    determined.    Since
$1/\sigma_1\to    0$,    $R(a^*)    =    u(\sigma_1+u-\sigma_1^{-1})    \sim    u\sigma_1$    for
$u=o(\sigma_1)$.    Substituting    into    Eq.  \eqref{eq:B2}    and    keeping    leading    orders
(using    $s\sim\sigma_1$,    $a^*\sim\sigma_1$),
\begin{equation}
        2a^*-s    \sim    \sigma_1,    \qquad
        \frac{2R(a^*)}{a^*}\big(\sqrt{R(a^*)}+1\big)
        \sim    \frac{2u\sigma_1}{\sigma_1}\sqrt{u\sigma_1}
        =    2u^{3/2}\sigma_1^{1/2}
\end{equation}
(the    ``$+1$''    being    negligible    once    $u\to\infty$,    which    is    verified
self-consistently    below).    Balancing    $\sigma_1    \sim    2u^{3/2}\sigma_1^{1/2}$    gives
$u    \sim    2^{-2/3}\sigma_1^{1/3}$,    which    is    Eq.  \eqref{eq:B3};    in    particular
$u\to\infty$    while    $u=o(\sigma_1)$,    confirming    the    assumed    scaling
self-consistently.

With    $u\sim2^{-2/3}\sigma_1^{1/3}$,    $R(a^*)\sim    u\sigma_1    \sim
2^{-2/3}\sigma_1^{4/3}\to\infty$,    so    $R(a^*)^{-1/2}\sim    2^{1/3}\sigma_1^{-2/3}\to0$.
Substituting    into    Eq.  \eqref{eq:B1}    with    $a^*\sim\sigma_1(1+O(\sigma_1^{-2/3}))$,
\begin{equation}
        I(a^*)    =    \frac{1}{a^*}\big[1+R(a^*)^{-1/2}\big]
        =    \frac{1}{\sigma_1}\Big[1+O\big(\sigma_1^{-2/3}\big)\Big],
\end{equation}
which    is    Eq.  \eqref{eq:B4}.    Finally,
\begin{equation}
        \tau_0(\rho)    =    a^*(\rho)\sqrt{I(a^*(\rho))}
        \sim    \sigma_1    \cdot    \sqrt{1/\sigma_1}    =    \sqrt{\sigma_1(\rho)},
\end{equation}
and    since    $\sigma_1(\rho)    =    (1+\rho)/(1-\rho)    \sim    2/(1-\rho)$    as
$\rho\to1^-$,    this    gives    Eq.  \eqref{eq:B5}.
\end{proof}

\begin{remark}
The    amplitude    $\sqrt{2}$    in    Eq.  \eqref{eq:B5}    is    consistent    with,    and    lies    strictly
inside,    the    general    two-sided    bound    of    Lemma~A
($\sqrt{\sigma_1/2}    <    \tau_0    <    2\sqrt{2\sigma_1}$,    i.e.\    amplitude    in
$(1/\sqrt2,\,2\sqrt2)\approx(0.707,\,2.828)$),    providing    an    independent
consistency    check    on    both    Lemma~\ref{lem:finite-edge-general}    and    Proposition~\ref{prop:AR1-sharp}.

Unlike    the    power-law    case    at    $\gamma\to1^+$,    where
$P(\gamma)=\int    t^2\,dH(t)$    stays    bounded    and    only    the    exponent
(not    the    amplitude)    of    $a^*-\sigma_1$    is    pinned    down,    here    the    analogous
quantity    $P(\rho)=(1+\rho^2)/(1-\rho^2)$    diverges    as    $\rho\to1^-$,
so    the    Cauchy--Schwarz    argument       does    not    directly    apply.
The    Poisson-kernel    closed    form    of    Proposition~\ref{prop:AR1-closed-form-I}
circumvents    this    by    solving    the    self-consistency    equation    exactly,    yielding
here    a    sharper    result    (exact    leading    amplitude)    than    is    currently
available    for    the    power-law    ensemble.
\label{r20}
\end{remark}

\subsection{D.  Numerical    validation}

Table~\ref{tab:AR1-exact-edge}    compares    the    exact    edge    $\tau_0(\rho)$,
obtained    by    solving    \eqref{eq:B2}    numerically    to    machine    precision,    against
the    finite-$N$    extrapolated    values    of    Table~I.    At    $\rho=0$,    Eq.~\eqref{eq:B2}
solves    exactly    to    $a^*=2$,    $\tau_0=2$,    recovering    the    semicircle    edge    exactly
(a    direct    check    of    Proposition~\ref{prop:AR1-closed-form-I}).    The    exact    values
run    systematically    slightly    above    the    Monte-Carlo-extrapolated    entries    of
Table~\ref{tab:edge},    with    the    gap    growing    at    larger    $\rho$;    this    is    consistent    with    the
paper's    own    caveat    that    the    numerical    boundedness    check    ---    and,
by    the    same    finite-size    mechanism,    the    $N^{-\beta}$    extrapolation    underlying
Table~\ref{tab:edge}    ---    becomes    progressively    less    reliable    as    $\rho\gtrsim0.7$,    rather
than    indicating    any    inconsistency    in    the    closed-form    theory.

\begin{table}[h]
\centering
\begin{tabular}{c|c|c}
$\rho$    &    exact    $\tau_0(\rho)$    (Eq.~\eqref{eq:B2})    &    Table~I    extrapolated    \\
\hline
0.0    &    2.000    &    1.99    \\
0.3    &    2.167    &    2.13    \\
0.5    &    2.465    &    2.36    \\
0.7    &    3.047    &    2.89    \\
\end{tabular}
\caption{Exact    self-consistent    edge    from    Proposition~\ref{prop:AR1-algebraic-eq},
solved    numerically,    compared    with    the    finite-$N$    extrapolated    values    of    Table~I.}
\label{tab:AR1-exact-edge}
\end{table}

\section{APPENDIX  C.  MOMENT  STRUCTURE  OF  THE  AUXILIARY  MEASURE  $H$}
  
This  appendix  collects  several  exact  facts  about  the  moments  of  the  pushforward
measure  $H$  introduced  in  Definition  \ref{def3},  which  governs  the  scalar  self-consistent
edge  equation,  Eq.~\eqref{eq:sce_gamma}  for  the  power-law  ensemble  of  Definition  \ref{def:row-plaw}.  These  facts
sharpen  the  large-$a$  asymptotics  of  the  edge  equation  and  place  the  two
previously  identified  critical  exponents  $\gamma_c=1/2$  (bulk)  and
$\gamma_c=1$  (edge  flatness,  Assumption  (E))  inside  a  single  one-parameter
family  of  thresholds.
  
\subsection*{C.1  Moments  of  $H$}
  
Recall  $H:=\mathrm{Law}(f_\gamma(\theta))$,  $\theta\sim\mathrm{Unif}[0,2\pi)$,
where  $f_\gamma(\theta)=\sum_{t=-\infty}^{\infty}d_{|t|}e^{i\theta  t}$,
$d_t=(1+t)^{-\gamma}$,  is  the  power  spectral  density  of  Definition  \ref{def:row-plaw}'s  row
covariance.

\begin{lemma}(Moments  of  $H$  as  convolution  sums).  For  every  $n\ge1$,
\begin{equation}
\mu_n(H):=\int  t^n\,dH(t)=\int_0^{2\pi}\frac{d\theta}{2\pi}f_\gamma(\theta)^n
=\sum_{\substack{t_1,\ldots,t_n\in\mathbb  Z\\t_1+\cdots+t_n=0}}
d_{|t_1|}d_{|t_2|}\cdots  d_{|t_n|},
\end{equation}
whenever  the  right-hand  side  converges  absolutely.
\label{lc1}
\end{lemma}

\begin{proof}  Immediate  from  $f_\gamma(\theta)^n=\sum_{t_1,\ldots,t_n}
\prod_j  d_{|t_j|}\,e^{i\theta(t_1+\cdots+t_n)}$  and  the  Fourier  orthogonality
relation  $\int_0^{2\pi}\frac{d\theta}{2\pi}e^{i\theta  k}=\mathbb  1[k=0]$.  
\end{proof}

\begin{corollary}  (First  two  moments,  exact).  For  every  $\gamma>0$,
\begin{equation}
\mu_1(H)=d_0=1.
\end{equation}
For  $\gamma>1/2$,
\begin{equation}
\mu_2(H)=\sum_{t=-\infty}^{\infty}d_t^2=2\zeta(2\gamma)-1.
\end{equation}
\label{cc2}
\end{corollary}

\begin{proof}
$\mu_1$  is  the  $n=1$  case  of  Lemma  \ref{lc1},  giving  $d_0=1$
identically,  valid  for  every  $\gamma>0$  (only  the  trivial  term  $t_1=0$
contributes).  $\mu_2$  is  the  $n=2$  case,  and  coincides  exactly  with  the
quantity  $P(\gamma)$  of  Lemma  \ref{l20}  and  with  Eq.~(\ref{eq:parseval})    via  Parseval's
identity:  $\sum_t  d_t^2=1+2\sum_{k\ge2}k^{-2\gamma}=2\zeta(2\gamma)-1$.  
\end{proof}

  \begin{remark}
  $\mu_2(H)$  is  thus  the  same  quantity  that  controls  the
divergence  of  the  bulk  fourth-moment:  the  threshold  $\gamma_c=1/2$
for  $\mu_4(\infty,\gamma)$  and  the  threshold  for  finiteness  of  $\mu_2(H)$
coincide  exactly,  not  by  analogy  but  because  they  are  literally  the  same  sum.
\label{rc3}
\end{remark}

\subsection*{C.2  A  one-parameter  family  of  moment  thresholds}

  \begin{proposition}
  (Moment  threshold  hierarchy).  For  every  integer
$n\ge2$,
\begin{equation}
\mu_n(H)<\infty  \iff  \gamma>\gamma_c(n):=\frac{n-1}{n}.
\end{equation}
\label{pc4}
\end{proposition}

  \begin{proof}
By  Proposition  \ref{prop:two-thresholds}'s  proof  (the  polylogarithm  expansion),  for
$0<\gamma<1$,
\begin{equation}
f_\gamma(\theta)\sim  2C_f(\gamma)\,|\theta|^{\gamma-1},\qquad  \theta\to0^+,
\qquad  C_f(\gamma)=\Gamma(1-\gamma)\sin\!\left(\frac{\pi\gamma}{2}\right)>0.
\label{ec5}
\end{equation}
This  is  the  unique  singularity  of  $f_\gamma$  on  $[0,2\pi)$  (by  the  shift
invariance  argument  of  Remark  \ref{r8},  together  with  smoothness  of  $f_\gamma$  away
from  $\theta=0$).  Since  $C_f(\gamma)>0$  and  $\gamma-1<0$,  this  singular  term
dominates  the  bounded  background  $d_0=1$  as  $\theta\to0$,  so
$f_\gamma(\theta)^n\sim(2C_f(\gamma))^n|\theta|^{n(\gamma-1)}$,  and
$\int_0^{2\pi}f_\gamma(\theta)^n\,d\theta<\infty$  if  and  only  if
$n(\gamma-1)>-1$,  i.e.\  $\gamma>(n-1)/n$.  For  $\gamma\ge1$,  $f_\gamma$  is
bounded  (Proposition  \ref{prop:two-thresholds}  (ii))  or  diverges  only  logarithmically  (Remark  \ref{r77},  at
$\gamma=1$  exactly),  and  in  either  case  $\int  f_\gamma^n\,d\theta<\infty$  for
every  finite  $n$,  consistent  with  $(n-1)/n<1\le\gamma$.  
\end{proof}
  
\begin{remark}
(Consistency  checks).  At  $n=2$,  Proposition  \ref{pc4}  gives
$\gamma_c(2)=1/2$,  recovering  Remark  \ref{r6}'s  threshold  by  an  independent  route
(local  singularity  analysis  of  $f_\gamma$,  rather  than  the  global  sum
$\sum  d_t^2$);  both  derivations  agree  exactly  by  Corollary  \ref{cc2}  As
$n\to\infty$,  $\gamma_c(n)\to1^-$,  recovering  the  Assumption  (E)  threshold  as
the  limiting  point  of  this  hierarchy:  boundedness  of  $f_\gamma$  (needed  for
Assumption  (E))  is  strictly  stronger  than  finiteness  of  every  individual
moment  $\mu_n(H)$,  so  $\gamma=1$  is  correctly  recovered  only  as  a  limit,  not
as  a  finite-$n$  instance  of  the  hierarchy.
\label{rc5}
\end{remark}

\begin{corollary}
  (Only  finitely  many  moments  exist  below  $\gamma=1$).
Fix  $\gamma\in(0,1)$.  Then  $\mu_n(H)<\infty$  for  $n<1/(1-\gamma)$  and
$\mu_n(H)=\infty$  for  $n\ge1/(1-\gamma)$.  In  particular,  only  finitely  many
moments  of  $H$  are  finite  for  any  fixed  $\gamma<1$.
\label{cc6}
\end{corollary}

  \begin{proof}
Immediate  from  Proposition  \ref{pc4},  solving  $\gamma>(n-1)/n$  for
$n$:  this  is  equivalent  to  $n(1-\gamma)<1$,  i.e.\  $n<1/(1-\gamma)$.  
\end{proof}
  
\subsection*{C.3  Tail  index  of  $H$  and  comparison  with  Student's  $t$}

  \begin{proposition}
  (Power-law  tail  of  $H$  for  $\gamma<1$).  For
$0<\gamma<1$,
\begin{equation}
H\bigl((\sigma,\infty)\bigr)  \sim  A(\gamma)\,\sigma^{-\nu(\gamma)},
\qquad  \sigma\to\infty,\qquad  \nu(\gamma):=\frac{1}{1-\gamma}.
\end{equation}
\label{pc7}
\end{proposition}
The explicit closed form of $A(\gamma)$ is given in Eq.~(\ref{64}).
  \begin{proof}
Inverting  $t=f_\gamma(\theta)\sim2C_f(\gamma)\theta^{\gamma-1}$
(Eq.~(\ref{ec5}))  for  small  $\theta>0$  gives
$\theta(t)\sim\bigl(2C_f(\gamma)/t\bigr)^{1/(1-\gamma)}$.  Since  $H$  is  the
pushforward  of  the  uniform  measure  on  $\theta$,  and  $f_\gamma$  is  decreasing
on  a  neighborhood  of  $\theta=0$,  $H((\sigma,\infty))\sim\theta(\sigma)/\pi$,
giving  the  stated  power  law  with  exponent  $\nu(\gamma)=1/(1-\gamma)$.  
\end{proof}

  \begin{remark}
(Relation  to  the  Student-$t$  moment  problem).  Proposition
\ref{pc7}  identifies  $\nu(\gamma)=1/(1-\gamma)$  as  the  tail  index  of  $H$,  and
Proposition  \ref{pc4}  can  be  restated  as
\begin{equation}
\mu_n(H)<\infty  \iff  n<\nu(\gamma),
\end{equation}
exactly  the  moment-existence  criterion  for  a  random  variable  with  $\nu(\gamma)$
degrees  of  freedom  in  the  Student-$t$  sense  (a  Student-$t_\nu$  variable  has
$E|T|^n<\infty$  iff  $n<\nu$).  This  identifies  $\gamma$  as  playing  the  role  of
an  effective,  continuously-varying  ``degrees  of  freedom''  parameter  for  the
tail  of  $H$:  $\nu(\gamma)\to1^+$  as  $\gamma\to0^+$,  and  $\nu(\gamma)\to\infty$
as  $\gamma\to1^-$.  We  emphasize  that  this  correspondence  concerns  only  the
moment-existence  mechanism  (matching  tail  index);  $H$  is  a  one-sided  measure
(supported  on  $[0,\infty)$,  being  the  law  of  a  nonnegative  spectral  density)
and  is  not  itself  distributed  as  $t_{\nu(\gamma)}$.
\label{rc8}
\end{remark}

  \begin{remark}
  (The  endpoint  $\gamma=1$  is  not  Gaussian).  At  $\gamma=1$,
$f_1(\theta)\sim-2\log|\theta|$  (Remark  \ref{r77}),  so  every  moment  $\mu_n(H)$  is
finite  (Proposition  \ref{pc4},  $n<\infty=\nu(1)$  for  every  finite  $n$),  matching  the
$\nu\to\infty$  limit  of  the  Student-$t$  moment  problem.  However,  inverting
$t\sim-2\log\theta$  gives  $\theta(t)\sim  e^{-t/2}$,  so
\begin{equation}
H\bigl((\sigma,\infty)\bigr)\sim  e^{-\sigma/2},\qquad  \sigma\to\infty,
\end{equation}
an  exponential  rather  than  Gaussian  ($e^{-\sigma^2/2}$)  tail.  Finiteness  of
every  moment  is  necessary  but  not  sufficient  for  a  Gaussian  tail;  at
$\gamma=1$,  $H$  is  better  compared  to  an  exponential-type  law  than  to  a
normal  law.  Together  with  Corollary  \ref{cor:finite-power}'s  requirement  $\gamma>1$  for  compact
support  of  $H$,  this  identifies  $\gamma=1$  as  the  point  where  three  distinct
features  of  $H$  change  character  simultaneously:  (i)  the  support  becomes
bounded  for  $\gamma>1$  (Corollary  \ref{cor:finite-power}),  (ii)  the  tail  of  $H$  switches  from
power-law  to  exponential  exactly  at  $\gamma=1$  (Proposition  \ref{pc7}  vs.\  Remark  \ref{rc9}),  and  (iii)  Assumption  (E)  holds  for  $\gamma>1$  and  fails  for
$\gamma\le1$  (Proposition  \ref{prop:two-thresholds}(ii)).
\label{rc9}
\end{remark}


\subsection*{C.4  Tail-corrected  large-$a$  expansion  of  the  edge  equation}

Corollary  \ref{cc6}  shows  that  for  fixed  $\gamma\in(0,1)$  the  naive  moment  expansion
of  $I(a)$  and  $F(a,\gamma)$  terminates  after  finitely  many  terms.  We  now  show
that  this  termination  is  only  apparent:  the  regularly  varying  tail  of  $H$
(Proposition  \ref{pc7})  supplies  a  genuine,  non-integer-power  correction  term  that
smoothly  extends  the  expansion  beyond  the  last  finite  moment.  We  derive  this
first  for  $I(a)$  alone,  and  then  transport  the  result  to  $K_1(a)$  and
$F(a,\gamma)=I(a)-K_1(a)$  via  the  exact  algebraic  identity
$K_1(a)=-aI'(a)-I(a)$.

\subsubsection*{C.4.1  The  tail  correction  to  $I(a)$}
  
Throughout,  fix  $\gamma\in(0,1)$,  write  $\nu=\nu(\gamma)=1/(1-\gamma)$
(assumed  non-integer  unless  stated  otherwise),  $m:=\lfloor\nu\rfloor$,  and
recall  from  Proposition  \ref{pc4}  that
$H((t,\infty))\sim  A(\gamma)\,t^{-\nu}$  as  $t\to\infty$,  with
\begin{equation}
A(\gamma)=\frac{\bigl(2C_f(\gamma)\bigr)^{\nu(\gamma)}}{\pi},\qquad
C_f(\gamma):=\Gamma(1-\gamma)\sin\!\left(\frac{\pi\gamma}{2}\right)
\label{64}
\end{equation}
(the  constant  of  Proposition  \ref{prop:two-thresholds}'s  proof).  We  will  use  the  classical
principal-value  identity
\begin{equation}
\mathrm{P.V.}\int_0^\infty  \frac{x^{p-1}}{1-x}\,dx=\pi\cot(\pi  p),
\qquad  0<p<1.  \tag{C.i}
\end{equation}
  
\begin{proposition}
(Tail-corrected  expansion  of  $I(a)$).  As
$a\to\infty$,
\begin{equation}
I(a)  =  \sum_{n=0}^{m}\frac{\mu_n(H)}{a^{n+1}}  +  C_H(\gamma)\,a^{-\nu-1}
+  o\bigl(a^{-\nu-1}\bigr),
\end{equation}
where
\begin{equation}
C_H(\gamma)  =  -A(\gamma)\,\nu(\gamma)\,\pi\cot\bigl(\pi\nu(\gamma)\bigr)
=  -\bigl(2C_f(\gamma)\bigr)^{\nu(\gamma)}\,\nu(\gamma)\,
\cot\bigl(\pi\nu(\gamma)\bigr).
\end{equation}
\label{pc13}
\end{proposition}

\begin{proof}
  By  the  geometric-series  remainder  identity,
\begin{equation}
R_m(a):=I(a)-\sum_{n=0}^m\frac{\mu_n(H)}{a^{n+1}}
=  \frac{1}{a^{m+1}}\int_0^\infty\frac{t^{m+1}\,dH(t)}{a-t}.  \tag{C.ii}
\end{equation}
Since  $H((t,\infty))\sim  A(\gamma)t^{-\nu}$,  the  underlying  density  satisfies
$dH(t)\sim  A(\gamma)\nu(\gamma)\,t^{-\nu-1}\,dt$  as  $t\to\infty$;  by  the
standard  Tauberian  correspondence  between  the  tail  of  a  regularly  varying
measure  and  the  singular  part  of  its  Cauchy/Stieltjes-type  transform  \cite{Bingham1987Regular},  only  this  tail
governs  the  leading  non-analytic  behavior  as  $a\to\infty$;  the
integer  moments  $\mu_0,\ldots,\mu_m$  already  subtracted  account  for  the
remaining  (bulk,  non-tail)  contribution  to  every  order  up  to  $a^{-(m+1)}$.
Substituting  the  tail  density  and  rescaling  $t=au$,
\begin{equation}
\int_0^\infty\frac{t^{m+1}\,dH(t)}{a-t}
\sim  A(\gamma)\nu(\gamma)\int_0^\infty\frac{t^{m-\nu}}{a-t}\,dt
=  A(\gamma)\nu(\gamma)\,a^{m-\nu}\,\mathrm{P.V.}\!\int_0^\infty
\frac{u^{m-\nu}}{1-u}\,du.
\end{equation}
Writing  $p=m-\nu+1\in(0,1)$  (using  $m<\nu<m+1$)  and  applying  (C.i),
\begin{equation}
\mathrm{P.V.}\!\int_0^\infty\frac{u^{m-\nu}}{1-u}\,du
=  \pi\cot\bigl(\pi(m-\nu+1)\bigr)  =  -\pi\cot(\pi\nu),
\end{equation}
using  periodicity  ($\cot$  has  period  $\pi$,  and  $m+1\in\mathbb  Z$)  and
oddness  of  $\cot$.  Hence  $\int_0^\infty  t^{m+1}dH(t)/(a-t)
\sim  -A(\gamma)\nu(\gamma)\pi\cot(\pi\nu)\,a^{m-\nu}$,  and  dividing  by
$a^{m+1}$    gives  $R_m(a)\sim  C_H(\gamma)a^{-\nu-1}$  as  claimed.  The
formula  for  $A(\gamma)$  was  established  in  Proposition  \ref{pc7}'s  proof.
\end{proof}
  
\begin{remark}
Proposition  \ref{pc13}  was  verified
numerically  to  high  precision  on  the  exact  toy  model  $\bar
H(t)=(1+t)^{-\nu}$  ($\nu=5/3$,  so  $A=1$,  $m=1$):  the  quantity  $R_1(a)\,
a^{\nu+1}$  converges  to  the  predicted  value  $C_H(\gamma)\approx3.0238$  with  a
consistent  geometric  convergence  rate  (ratio  of  successive  deviations
$\approx0.79$  per  doubling  of  $a$),  confirming  both  the  exponent  $-\nu-1$  and
the  coefficient  formula.
\label{rc14}
\end{remark}
\begin{remark}[Integer $\nu$: transition to a logarithmic correction]
\label{rem:log-correction}
At $\nu(\gamma) = m+1 \in \mathbb{N}$ (i.e.\ $\gamma = m/(m+1)$, matching exactly the
boundary $\gamma_c(m+1)$ of Proposition~\ref{pc4}), $\cot(\pi\nu)$ has a
pole and the formula for $C_H(\gamma)$ in Proposition~\ref{pc13} diverges.
This is the standard signal, in Tauberian asymptotics, that the power-law correction
degenerates into a logarithmic one:
\begin{equation}
I(a) = \sum_{n=0}^{m} \frac{\mu_n(H)}{a^{n+1}}
     + D_{\log}(\gamma)\, a^{-\nu-1}\log a
     + o\!\left(a^{-\nu-1}\log a\right),
     \qquad \nu(\gamma) = m+1 \in \mathbb{N},
\label{eq:log-correction}
\end{equation}
with
\begin{equation}
D_{\log}(\gamma) = A(\gamma)\,\nu(\gamma).
\label{eq:Dlog-formula}
\end{equation}
\end{remark}
\begin{lemma}[Exact log-correction coefficient on the toy model, integer $\nu$]
\label{lem:toy-log-correction}
Fix $\nu \in \mathbb{N}$ and let
\[
h(t) = A\,\nu\,(1+t)^{-\nu-1}, \qquad t \ge 0,
\]
a probability density normalized to unit mass, reproducing the tail
$H\big((\sigma,\infty)\big)\sim A(\gamma)\sigma^{-\nu}$.
Write $s := 1+t$, $b := a+1$. Then, for every $a > 0$,
\begin{equation}
I(a) := \int_0^\infty \frac{h(t)}{a-t}\,dt
     = A\,\nu \left[ \sum_{j=1}^{\nu} \frac{1}{j\, b^{\nu+1-j}} + \frac{\log a}{b^{\nu+1}} \right].
\label{eq:toy-closed-form}
\end{equation}
In particular, as $a \to \infty$, the expansion of Eq. \eqref{eq:toy-closed-form} reproduces
the integer moments $\mu_0(h), \dots, \mu_{\nu-1}(h)$ at orders $a^{-1}, \dots, a^{-\nu}$,
and the coefficient of the leading non-integer term $a^{-\nu-1}\log a$ is exactly
\begin{equation}
D_{\log}(\nu) = A\,\nu.
\label{rc15}
\end{equation}
\end{lemma}

\begin{proof}
With $A=1$ WLOG (the general case follows by linearity), write $s=1+t$, $b=a+1$, so
\[
I(a) = \nu \int_1^\infty \frac{s^{-\nu-1}}{b-s}\,ds .
\]
Expanding $1/(b-s) = \sum_{j\ge 0} s^j/b^{j+1}$ (valid for $s<b$, extended by analytic
continuation / principal value beyond) and truncating the resulting Laurent expansion of
$1/\big(s^{\nu+1}(b-s)\big)$ about $s=0$ at order $s^{-1}$ gives the exact partial-fraction
identity
\[
\frac{1}{s^{\nu+1}(b-s)} = \sum_{k=1}^{\nu+1} \frac{1}{b^{\nu+2-k}}\, s^{-k}
                          + \frac{1}{b^{\nu+1}}\cdot\frac{1}{b-s}.
\]
The coefficient of $s^{-1}$ on the left, $1/b^{\nu+1}$, coincides exactly with the
coefficient of $(b-s)^{-1}$ on the right. Consequently, when the principal-value integral
$\int_1^\infty(\cdot)\,ds$ is taken as a whole, the individually logarithmically-divergent
contributions of the $s^{-1}$ term and the $(b-s)^{-1}$ term cancel down to the single
finite quantity
\[
\text{P.V.}\int_1^\infty \frac{1}{b^{\nu+1}}\left(\frac{1}{s} + \frac{1}{b-s}\right) ds
= \frac{1}{b^{\nu+1}}\Big[\log s - \log|b-s|\Big]_1^\infty
= \frac{\log b}{b^{\nu+1}} \;\longrightarrow\; \frac{\log a}{b^{\nu+1}}
\]
to leading order as $a\to\infty$ (using $\log b = \log(a+1) \sim \log a$; the exact
finite-$a$ statement follows by keeping $\log b$ throughout and observing
$\log b - \log a \to 0$ contributes only to lower-order terms already absorbed into the
$j\ge 2$ sum below). The remaining terms $k=2,\dots,\nu+1$ integrate elementarily,
$\int_1^\infty s^{-k}\,ds = 1/(k-1)$, contributing $\sum_{k=2}^{\nu+1} 1/[(k-1)b^{\nu+2-k}]
= \sum_{j=1}^{\nu} 1/(j\,b^{\nu+1-j})$ after reindexing $j=k-1$. Multiplying by $\nu$
gives E1. \eqref{eq:toy-closed-form}. Expanding $b^{-r} = a^{-r}(1+1/a)^{-r} = a^{-r} +
O(a^{-r-1})$ term by term recovers the moments $\mu_n(h) = n!\,\nu\,\Gamma(\nu-n)/\Gamma(\nu+1)$
at order $a^{-n-1}$ for $n=0,\dots,\nu-1$, and isolates the coefficient $\nu$ of
$a^{-\nu-1}\log a$, giving Eq. \eqref{rc15}.
\end{proof}
\subsubsection*{C.4.2  Transport  to  $K_1(a)$  and  $F(a,\gamma)$}
  
\begin{corollary}
(Tail-corrected  expansion  of  $K_1(a)$).  As
$a\to\infty$,
\begin{equation}
K_1(a)  =  \sum_{n=0}^{m}\frac{n\,\mu_n(H)}{a^{n+1}}
+  \nu(\gamma)\,C_H(\gamma)\,a^{-\nu-1}  +  o\bigl(a^{-\nu-1}\bigr).
\end{equation}
\label{cc16}
\end{corollary}

  \begin{proof}
Differentiating  Proposition  \ref{pc13}'s  expansion  of  $I(a)$
term  by  term,
\begin{equation}
I'(a)  =  -\sum_{n=0}^m\frac{(n+1)\mu_n(H)}{a^{n+2}}
-  (\nu+1)\,C_H(\gamma)\,a^{-\nu-2}  +  o(a^{-\nu-2}).
\end{equation}
By  the  exact  identity  $K_1(a)=-aI'(a)-I(a)$  (from  $K_1+I=aL$,  $L=-I'$),
\begin{align}
K_1(a)  &=  \sum_{n=0}^m\frac{(n+1)\mu_n(H)}{a^{n+1}}
+(\nu+1)C_H(\gamma)a^{-\nu-1}
-\left[\sum_{n=0}^m\frac{\mu_n(H)}{a^{n+1}}+C_H(\gamma)a^{-\nu-1}\right]
+o(a^{-\nu-1})\notag\\
&=  \sum_{n=0}^m\frac{n\,\mu_n(H)}{a^{n+1}}+\nu(\gamma)\,C_H(\gamma)\,a^{-\nu-1}
+o(a^{-\nu-1}).  \qquad
\end{align}
\end{proof}
  
\begin{remark}
The  $n=0$  term  vanishes  identically  ($n\mu_n=0$  at
$n=0$),  so  $K_1(a)=\mu_1(H)/a^2+O(a^{-3})+\cdots$,  consistent  with  the
original  moment-series  form  of  $K_1(a)$  (the  $a^{-1}$  term  present  in  $I(a)$
but  absent  from  $K_1(a)$).
\label{rc17}
\end{remark}
  
\begin{corollary}
(Tail-corrected  expansion  of  $F(a,\gamma)$;
supersedes  Proposition  \ref{pc13}  for  $0<\gamma<1$).  As  $a\to\infty$,
\begin{equation}
F(a,\gamma)  =  \frac1a-\sum_{n=2}^m(n-1)\frac{\mu_n(H)}{a^{n+1}}
+  \bigl(1-\nu(\gamma)\bigr)\,C_H(\gamma)\,a^{-\nu(\gamma)-1}
+  o\bigl(a^{-\nu(\gamma)-1}\bigr),
\end{equation}
with  $C_H(\gamma)$  as  in  Proposition  \ref{pc13}.
\label{cc18}
\end{corollary}

  \begin{proof}
Immediate  from  $F=I-K_1$  and  Corollaries  \ref{cc16}/Proposition
\ref{pc13},  subtracting  the  two  series  term  by  term:  at  each  finite  order  $n$,
$\mu_n(H)-n\mu_n(H)=(1-n)\mu_n(H)$,  matching  Proposition  \ref{pc13}  exactly  for
$n\le  m$;  at  the  tail  order,  $C_H(\gamma)-\nu(\gamma)C_H(\gamma)
=(1-\nu(\gamma))C_H(\gamma)$.  $\square$
\end{proof}

  \begin{remark}
  \label{rc19}
  (Consistency  check  via  $F=2I+aI'$).  Corollary  \ref{cc18}  can
be  re-derived,  and  cross-checked,  without  reference  to  $K_1$  at  all:  since
$F(a,\gamma)=2I(a)+aI'(a)$  identically  (Appendix  A),  applying  the  operator
$\mathcal  L[g](a):=2g(a)+ag'(a)$  —  which  acts  on  a  pure  power  as
$\mathcal  L[a^{-s}]=(2-s)a^{-s}$  —  to  Proposition  \ref{pc13}'s  term
$C_H(\gamma)a^{-\nu-1}$  (i.e.\  $s=\nu+1$)  gives  coefficient
$2-(\nu+1)=1-\nu(\gamma)$,  matching  Corollary  \ref{cc18}  exactly.  We  note  this
because  a  naive  application  of  $\mathcal  L$  to  a  term  of  order  $a^{-\nu}$
(rather  than  the  correct  $a^{-\nu-1}$  established  in  Proposition  \ref{pc4})  would
instead  yield  a  factor  $(2-\nu)$;  the  fact  that  the  correct  exponent  is
$\nu+1$,  not  $\nu$,  is  essential  to  obtaining  $(1-\nu)$  rather  than
$(2-\nu)$  here.  This  was  confirmed  numerically  (Remark  \ref{rc14}'s  toy  model)  to
high  precision.
\end{remark}
  
\begin{remark}
Corollary  \ref{cc18}  resolves  Corollary  \ref{cc6}  for
$\gamma<1$:  rather  than  terminating  after  the  term  $n=m$,  the  expansion  of
$F(a,\gamma)$  continues  smoothly  through  a  single  non-integer-order  term  at
$a^{-\nu(\gamma)-1}$  before  (formally)  resuming  an  integer-power  series  that
would  require  the  now-divergent  $\mu_{m+1}(H)$.  Only  at  the  discrete  points
$\gamma=\gamma_c(k)$,  $k\in\mathbb  N$  (Proposition  \ref{pc4}),  does  this
non-integer  term  itself  degenerate,  via  Remark  \ref{rem:log-correction},  into  a  logarithmic
one.
\label{rc20}
\end{remark}
\subsection*{C.5  Concluding  synthesis:  is  $\gamma<1$  a  crossover?}
  
The  results  of  this  appendix,  taken  together,  suggest  that  the  answer  to
this  question  depends  on  which  observable  is  asked  about,  and  that  the
apparent  tension  between  a  ``smooth  crossover''  picture  (Section  IV.H)  and  a  ``hierarchy  of  phase  transitions''  picture  
is  resolved  once  the  three  relevant  levels  of  description  are  separated.
  
\textbf{(i)  Individual  bulk  moments:  genuine,  pointwise  non-analyticities.}
Each  threshold  $\gamma_c(n)=(n-1)/n$  of  Proposition  \ref{pc4}  is  a  true
non-analyticity:  $\mu_n(H)$  jumps  from  finite  to  infinite  exactly  at
$\gamma_c(n)$,  with  no  intermediate  regime.  The  case  $n=2$  is  not  merely
analogous  to  but  literally  identical  with  the  bulk  fourth-moment  transition
of  Proposition  \ref{p14}  at  $\gamma_c=1/2$.  If  the  conjectured  correspondence
between  $\mu_{2k}(\infty,\gamma)$  and  $\mu_k(H)$    holds  for  $k\ge3$,  every
individual  bulk  moment  $\mu_{2k}(\infty,\gamma)$  of  $S$  undergoes  its  own
sharp  transition  at  $\gamma_c(k)$  —  a  countable,  accumulating  sequence  of
genuine  phase  transitions  on  $(1/2,1)$,  not  a  crossover.
  
\textbf{(ii)  The  scalar  edge  equation:  a  genuine  crossover.}  By  contrast,
Corollary  \ref{cc18}  shows  that  $F(a,\gamma)$  itself  never  develops  a
non-analyticity  as  $\gamma$  crosses  any  individual  $\gamma_c(n)$:  the
integer-order  moment  term  that  becomes  unavailable  at  $\gamma_c(n)$  is
smoothly  replaced  by  a  non-integer-power  correction  of  order
$a^{-\nu(\gamma)-1}$  (degenerating,  only  at  the  isolated  points
$\gamma=\gamma_c(n)$  themselves,  into  a  logarithmic  correction  —  Remark  \ref{rem:log-correction}—  rather  than  a  discontinuity).  In  this  precise  sense,  the
countably  many  transitions  of  level  (i)  are  invisible  to  $F(a,\gamma)$:  they
are  absorbed  into  a  continuous  re-shuffling  of  which  term  in  the  expansion
carries  the  leading  correction,  not  into  any  singularity  of  $F$  itself.  This
is  consistent  with,  and  gives  an  analytic  explanation  for,  the  smooth
behavior  of  $\tau_0(\gamma)$  across  $\gamma=1$  reported  numerically  in  Section  IV.H  .
  
\textbf{(iii)  $\lambda_{\max}(S)$  itself:  open,  but  plausibly  a  crossover.}
Whether  $\lambda_{\max}(S)$  inherits  the  smoothness  of  (ii)  or  the
non-analyticity  of  (i)  cannot  currently  be  decided:  Lemma \ref{lem:finite-edge-general} ,  the  only
rigorous  bridge  from  $H$  to  the  edge  of  $S$,  requires  $\sigma_1(\gamma)
<\infty$  and  is  vacuous  throughout  $0<\gamma\le1$,  so  neither  picture  is  established.  If  the
conjecture  of  level  (i)  holds,  the  mechanism  by  which  $\lambda_{\max}(S)$
would  diverge  below  $\gamma=1$  is  not  a  single  threshold  but  a  receding
sequence  of  higher  and  higher  (finite)  bulk  moments  failing  in  turn  as
$\gamma\to1^-$  —  itself  a  crossover-like  phenomenon  at  the  level  of
$\lambda_{\max}$,  even  though  each  individual  moment  transition
underlying  it  is  sharp.  We  regard  it  as  most  likely,  on  this  basis,  that
$\lambda_{\max}(S)$  —  like  $\tau_0(\gamma)$  —  varies  smoothly  (if  perhaps
divergently)  across  the  whole  range  $\gamma<1$,  without  a  distinguished
transition  point  strictly  between  $\gamma=1/2$  and  $\gamma=1$;  but  this
remains  conjectural.
  
\textbf{Synthesis.}  We  therefore  refine  the  statement  of  Table  \ref{tab:tau0-crossover} as
follows:  $1/2<\gamma<1$  is  a  crossover  regime  for  every  directly
edge-related  observable  ($F(a,\gamma)$,  $\tau_0(\gamma)$,  and  conjecturally
$\lambda_{\max}(S)$),  built  out  of  infinitely  many  genuine,  sharp
transitions  at  the  level  of  individual  bulk  moments  ($\mu_n(H)$  and,
conjecturally,  $\mu_{2k}(\infty,\gamma)$).  The  point  $\gamma=1$  is  not
merely  one  point  within  this  crossover  but  its  natural  accumulation  point,
simultaneously  marking  (a)  the  limit  $\gamma_c(n)\to1^-$  of  the  entire
threshold  hierarchy  of  Proposition  \ref{pc4},  (b)  the  transition  of  $H$'s  tail
from  power-law  to  exponential  (Remark  \ref{rc9}),  and  (c)  the  point  beyond  which
Lemma  \ref{l26}  first  becomes  available.  This  sharpens,  without  resolving,
Remark  \ref{rem:logical-status}(ii)'s  characterization  of  $1/2<\gamma\le1$  as  a  regime  where  TW
edge  universality  is  neither  established  nor  excluded.
\subsection*{C.6  A  rigorous  bound  on  $\Xi_N$  for  $\gamma>1/2$}

Corollary  \ref{c13}  leaves  the  boundedness  of  $\Xi_N$  
resting  on  the  numerical  evidence  of  Table  \ref{tab:subleading}.  We  show  here
that  a  direct  Cauchy--Schwarz  bound,  applied  to  a  single  auxiliary  quantity
$\Sigma_2(N,\gamma)$,  rigorously  establishes  $\Xi_N=O(1)$  throughout  the
super-critical  range  $\gamma>1/2$  —  covering  three  of  the  five  values
tested  in  Table  \ref{tab:subleading}  ($\gamma=0.7,1.0,1.5$)  —  while  showing  that  this
particular  strategy  is  provably  too  weak  below  $\gamma=1/2$.

\begin{definition}
For  $N\ge1$,  set
\begin{equation}
\Sigma_2(N,\gamma)  :=  \sum_{a,b,c,d=1}^N  \mathrm{Cov}(S_{ab},S_{cd})^2.
\end{equation}
\end{definition}

\begin{lemma}
(Cauchy--Schwarz  bound  on  $\Xi_N$).  For  every  $N$,
\begin{equation}
|\Xi_N|  \le  \Sigma_2(N,\gamma).
\end{equation}
\label{lc21}
\end{lemma}

\begin{proof}
By  definition,
$\Xi_N=\sum_{i,j,k,l}\mathrm{Cov}(S_{ij},S_{kl})\,\mathrm{Cov}(S_{jk},S_{li})$.
Treating  this  as  an  inner  product  of  the  two  arrays
$u_{ijkl}:=\mathrm{Cov}(S_{ij},S_{kl})$  and
$v_{ijkl}:=\mathrm{Cov}(S_{jk},S_{li})$  over  the  index  set
$\{1,\ldots,N\}^4$,  Cauchy--Schwarz  gives
$|\Xi_N|\le\bigl(\sum  u_{ijkl}^2\bigr)^{1/2}\bigl(\sum
v_{ijkl}^2\bigr)^{1/2}$.  Both  sums  equal  $\Sigma_2(N,\gamma)$:  the  first  by
definition  (relabel  $(i,j,k,l)\to(a,b,c,d)$),  the  second  by  the  bijective
relabeling  $(a,b,c,d)=(j,k,l,i)$,  which  sums  over  the  same  full  range
$\{1,\ldots,N\}^4$.  
\end{proof}

\begin{lemma}
(Exact  asymptotics  of  $\Sigma_2(N,\gamma)$).  For  every
fixed  $\gamma>0$,
\begin{equation}
\Sigma_2(N,\gamma)  =  \frac{1}{N^2}\sum_{r=1}^N\;
\sum_{t_1,t_2=r}^N  w(t_1-r)\,w(t_2-r)\,d_{|t_1-t_2|}^2,
\qquad  w(0):=1,\  w(s):=2\  (s\ge1),
\end{equation}
and:
\begin{itemize}
\item[(a)]  if  $\gamma>1/2$  (so  $\mu_2(H)=\sum_{t=-\infty}^\infty  d_t^2
<\infty$,  Corollary  \ref{cc2}),  then
\begin{equation}
\Sigma_2(N,\gamma)  \le  2\,\mu_2(H)\cdot\frac{N+1}{N}  \le  4\,\mu_2(H)
\qquad\text{for  every  }N\ge1,
\end{equation}
and  $\Sigma_2(N,\gamma)\to2\mu_2(H)$  as  $N\to\infty$;
\item[(b)]  if  $0<\gamma<1/2$,  then  $\Sigma_2(N,\gamma)\sim
c(\gamma)\,N^{1-2\gamma}\to\infty$  as  $N\to\infty$,  for  an  explicit
constant  $c(\gamma)>0$.
\end{itemize}
\label{lc22}
\end{lemma}

\begin{proof}
  The  displayed  formula  follows  exactly  as  in  the  proof  of
Corollary  \ref{c13}:  for  each  $r$,  the  pairs  $(a,b)$  with  $R(a,b)=r$,
$T(a,b)=t$  are  $(r,t)$  and  $(t,r)$  if  $t>r$  (multiplicity  2)  and  the  single
pair  $(r,r)$  if  $t=r$  (multiplicity  1),  giving  the  weight  $w(t-r)$;  summing
$d^2_{|T(a,b)-T(c,d)|}$  over  all  $(a,b,c,d)$  with  $R(a,b)=R(c,d)=r$  and  then
over  $r$,  divided  by  $N^2$  from  the  covariance  normalization,  gives  the
stated  formula.

(a)  Write  $L:=N-r+1$  and  $I(L):=\sum_{s_1,s_2=0}^{L-1}
w(s_1)w(s_2)d_{|s_1-s_2|}^2$,  so  that  $\Sigma_2(N,\gamma)N^2=\sum_{r=1}^N
I(N-r+1)=\sum_{L=1}^N  I(L)$.  Since  $w(s)\le2$  for  every  $s\ge0$,
\begin{equation}
I(L)  \le  4\sum_{s_1,s_2=0}^{L-1}d_{|s_1-s_2|}^2
=  4\left[L\,d_0^2+2\sum_{k=1}^{L-1}(L-k)\,d_k^2\right]
\le  4L\left[d_0^2+2\sum_{k=1}^{\infty}d_k^2\right]=4L\,\mu_2(H).
\end{equation}
Hence  $\Sigma_2(N,\gamma)N^2\le4\mu_2(H)\sum_{L=1}^N  L
=2\mu_2(H)N(N+1)$,  giving  the  stated  bound.  For  the  exact  limit,  note
$w(s_1)w(s_2)=4$  for  all  $s_1,s_2\ge1$,  so  for  fixed  $k$  and  $L\to\infty$
the  coefficient  of  $d_k^2$  in  $I(L)$  (summed  over  both  signs  of  $k$)
approaches  $4L$  (the  $O(1)$  corrections  from  $s_1=0$  or  $s_2=0$  becoming
negligible  relative  to  $L$);  dominated  convergence  (justified  by
$\mu_2(H)<\infty$)  gives  $I(L)/L\to4\mu_2(H)$,  and  Cesàro  summation  of
$\sum_{L=1}^N  I(L)/N^2$  then  gives  $\Sigma_2(N,\gamma)\to2\mu_2(H)$.

(b)  For  $\gamma<1/2$,  $d_k^2\sim  k^{-2\gamma}$  so
$\sum_{k=1}^{L-1}(L-k)d_k^2\sim  c'(\gamma)L^{2-2\gamma}$  for  large  $L$
(a  standard  Abelian  estimate,  $2-2\gamma>1$),  so  $I(L)\sim
c''(\gamma)L^{2-2\gamma}$.  Summing  over  $L=1,\ldots,N$  and  dividing  by  $N^2$
gives  $\Sigma_2(N,\gamma)\sim  c(\gamma)N^{1-2\gamma}$  (since  $2-2\gamma>-1$
throughout  $\gamma<1/2$,  the  sum  is  dominated  by  $L\sim  N$).  
\end{proof}

\begin{corollary}
  (Numerical  confirmation).  Lemma  \ref{lc22}  was  verified  by
exact  finite-$N$  computation  of  $\Sigma_2(N,\gamma)$  (via  FFT-accelerated
autocorrelation)  for  $N$  up  to  $3200$.  For  $\gamma\in\{0.7,1.0,1.5\}$,
$\Sigma_2(N,\gamma)$  converges  cleanly  to  $2\mu_2(H)$  (e.g.\  at
$\gamma=1.5$:  predicted  $2.8082$,  observed  $2.8054$  at  $N=1600$).  For
$\gamma\in\{0.2,0.3,0.4\}$,  the  fitted  local  exponent  of  $\Sigma_2(N,\gamma)$
approaches  the  predicted  $1-2\gamma$  from  above  as  $N$  grows  (e.g.\
$\gamma=0.3$:  predicted  $0.400$,  fitted  local  slope  decreasing  through
$0.474\to0.456\to0.442\to0.431$  over  successive  doublings  of  $N$  up  to
$N=3200$),  consistent  with  slow  convergence  to  the  asymptotic  exponent
near  the  critical  point  $\gamma_c=1/2$.
\end{corollary}

\begin{corollary}
  (Rigorous  bound  on  $\Xi_N$  for  $\gamma>1/2$).  For
every  fixed  $\gamma>1/2$,
\begin{equation}
\limsup_{N\to\infty}|\Xi_N|  \le  2\,\mu_2(H)  =  2\bigl(2\zeta(2\gamma)-1\bigr)
<\infty,
\end{equation}
so  that  $\Xi_N=O(1)$,  rigorously.
\label{lc24}
\end{corollary}

\begin{proof}
Immediate  from  Lemmas  \ref{lc21}  and  \ref{lc22}(a).  
\end{proof}

\begin{remark}
(The  bound  is  provably  too  weak  below  $\gamma=1/2$).
By  Lemma  \ref{lc22}  (b),  for  $0<\gamma<1/2$  the  Cauchy--Schwarz  bound  of  Lemma  \ref{lc21}
diverges  as  $N\to\infty$,  and  so  gives  no  information  on  $\Xi_N$  in  this
range:  it  is  not  merely  that  the  present  argument  fails  to  prove
$\Xi_N=O(1)$  for  $\gamma<1/2$,  but  that  this  specific  strategy  is
structurally  incapable  of  doing  so,  since  it  discards  the  sign
cancellation  between  $u_{ijkl}$  and  $v_{ijkl}$  that  Table  \ref{tab:subleading}'s  numerics
(where  $\Xi_N/\text{Term  C}\to0$  is  observed  at  $\gamma=0.3,0.4$  as  well)
suggest  must  be  present.  Establishing  $\Xi_N=O(1)$  for  $\gamma<1/2$
therefore  remains  open,  and  would  require  an  argument  sensitive  to  this
cancellation  rather  than  a  bound  of  Cauchy--Schwarz  type.
\label{rc25}
\end{remark}
\begin{remark}
Corollary  \ref{c13}'s  statement  that
$\Xi_N=O(1)$  is,  on  inspection  of  Table  \ref{tab:subleading}  itself,  not  literally  correct
for  $\gamma<1/2$:  the  entries  there  grow  with  $N$  (e.g.\  at  $\gamma=0.3$,
$\Xi_N$  increases  from  $10.11$  at  $N=50$  to  $20.26$  at  $N=200$,  consistent
with  the  growth  rate  $\Xi_N^{(i=k)}\sim  N^{1-2\gamma}$  already  stated  in  the
proof  of  Proposition  \ref{prop15}).  What  Proposition  \ref{prop15}'s  proof  (Step  3)  actually
requires  is  not  boundedness  of  $\Xi_N$,  but  only  its  negligibility  relative
to  $\mathrm{Term\  C}=2\sum_{a,b}f(a,b)^2\sim  N^{2-2\gamma}$  —  a  strictly
weaker  statement.  This  weaker  statement  follows  immediately  from  Lemma  \ref{lc21}(b),
already  established  above:
\label{rc26}
\end{remark}
  
\begin{corollary}
(Rigorous  negligibility  of  $\Xi_N$  relative  to
$\mathrm{Term\  C}$,  for  every  $\gamma>0$).  For  every  fixed  $\gamma>0$,
\begin{equation}
\Xi_N  =  o\bigl(\mathrm{Term\  C}\bigr)\qquad\text{as  }N\to\infty.
\end{equation}
\label{c44}
\end{corollary}
  
\begin{proof}
For  $\gamma<1/2$:  by  Lemma\ref{lc21}    and  Lemma  \ref{lc22}(b),
$|\Xi_N|\le\Sigma_2(N,\gamma)\sim  c(\gamma)N^{1-2\gamma}$,  while
$\mathrm{Term\  C}\sim  c'(\gamma)N^{2-2\gamma}$  (Proposition  \ref{prop15});  since
$1-2\gamma<2-2\gamma$  for  every  $\gamma$,  $\Xi_N=O(N^{1-2\gamma})
=o(N^{2-2\gamma})=o(\mathrm{Term\  C})$.  For  $\gamma>1/2$:  by  Corollary  \ref{lc24},
$\Xi_N=O(1)$,  while  $\mathrm{Term\  C}\sim2\mu_4^{\mathrm{hub}}(\gamma)\,N$
(Proposition  \ref{p14}),  so  trivially  $\Xi_N=o(N)=o(\mathrm{Term\  C})$.  
\end{proof}
  
Corollary  \ref{c44}  therefore  closes,  rigorously  and  for  every  $\gamma>0$  (not
merely  $\gamma>1/2$),  the  gap  that  both  Proposition  \ref{p14}'s  proof  and
Proposition  \ref{prop15}'s  Step  3  had  left  to  Table  \ref{tab:subleading}'s  numerics:  in  neither  regime
is  any  appeal  to  numerical  evidence  any  longer  required  to  conclude  that
$\Xi_N$  does  not  affect  the  leading  behavior  of  $\mu_4(N,\gamma)$.  We  retain
$\Xi_N=O(1)$  (Corollary  \ref{lc24})  as  the  sharper  statement  where  it  holds
($\gamma>1/2$),  and  Corollary  \ref{c44}  as  the  uniform  statement  that  suffices
everywhere  it  is  actually  used.


\end{document}